\documentclass[12pt, peerreview, compsocconf]{IEEEtran}

\usepackage[T1]{fontenc}
\usepackage[UKenglish]{babel}
\usepackage{url,xspace}
\usepackage[usenames,dvipsnames]{color}
\usepackage{stmaryrd}
\SetSymbolFont{stmry}{bold}{U}{stmry}{m}{n}
\usepackage{paralist}

\usepackage[pdftex]{graphicx}
\usepackage{wrapfig}

\graphicspath{{./pdfs/}}

\usepackage{amsthm,amssymb,amsmath}

\DeclareMathAlphabet{\mathcal}{OMS}{cmsy}{m}{n}

\usepackage{fixltx2e}

\usepackage{array}

\usepackage{tikz}

\usetikzlibrary{patterns}
\usetikzlibrary{calc,positioning}
\usetikzlibrary{shapes}
\usetikzlibrary{matrix}
\usetikzlibrary{decorations}
\usetikzlibrary{decorations.pathmorphing}
\usetikzlibrary{decorations.markings}
\usetikzlibrary{decorations.pathreplacing}
\usetikzlibrary{backgrounds}
\usepackage[margin=1in]{geometry}

\usepackage{hyperref}
\hypersetup{colorlinks=true, linkcolor=red, citecolor=blue}

\usepackage[linesnumbered,vlined,noend,ruled,nofillcomment]{algorithm2e}
\usepackage[justification=centering]{caption}

\usepackage{ifthen}
\usepackage{arrayjobx}

\usepackage{dsfont}

\newcommand{\ema}[1]{\ensuremath{#1}\xspace}

\newcommand{\inte}[2][0]{\ema{\left\llbracket #1,#2 \right\rrbracket}}

\newcommand{\pro}[1]{\ema{\mathds{P}\left(#1\right)}}
\newcommand{\expe}[1]{\ema{\mathds{E}\left(#1\right)}}

\newcommand{\indi}[2]{\ema{\mathds{1}_{#1}\left( #2 \right)}}

\newcommand{\expu}[1]{\ema{e^{#1}}}
\newcommand{\expi}[1]{\ema{\operatorname{exp}\left( #1\right)}}

\newcommand{\pinf}{\ema{+\infty}}

\makeatletter
\newsavebox\myboxA
\newsavebox\myboxB
\newlength\mylenA

\newcommand*\xoverline[2][0.75]{%
    \sbox{\myboxA}{$\m@th#2$}%
    \setbox\myboxB\null
    \ht\myboxB=\ht\myboxA%
    \dp\myboxB=\dp\myboxA%
    \wd\myboxB=#1\wd\myboxA
    \sbox\myboxB{$\m@th\overline{\copy\myboxB}$}
    \setlength\mylenA{\the\wd\myboxA}
    \addtolength\mylenA{-\the\wd\myboxB}%
    \ifdim\wd\myboxB<\wd\myboxA%
       \rlap{\hskip 0.5\mylenA\usebox\myboxB}{\usebox\myboxA}%
    \else
        \hskip -0.5\mylenA\rlap{\usebox\myboxA}{\hskip 0.5\mylenA\usebox\myboxB}%
    \fi}
\makeatother

\newcommand{\barov}[1]{\xoverline[1]{#1}}

\newtheorem{theorem}{Theorem}

\newtheorem{lemma}{Lemma}

\theoremstyle{definition}

\newcommand{\figsidebyside}[9]{
\begin{figure}[#1]
\centering
\begin{minipage}{#2\textwidth}
\begin{center}
\includegraphics[width=\textwidth]{#3}
\end{center}
\captionof{figure}{#4 \label{fig:#5}}
\end{minipage}\hfill%
\begin{minipage}{#6\textwidth}
\begin{center}
\includegraphics[width=\textwidth]{#7}
\end{center}
\captionof{figure}{#8 \label{fig:#9}}
\end{minipage}
\end{figure}
}

\makeatletter
\g@addto@macro\normalsize{%
  \setlength\abovedisplayskip{2pt}
  \setlength\belowdisplayskip{2pt}
  \setlength\abovedisplayshortskip{2pt}
  \setlength\belowdisplayshortskip{2pt}
}
\makeatother

\newcommand{\vspp}[1]{\pp{\vspace*{#1cm}}}

\newcommand{\sta}[1]{\ema{\mathcal{S}_{#1}}}

\newcommand{\eve}[1]{\ema{E_{\mathrm{#1}}}}
\newcommand{\thex}[2]{\ema{X_{[#1,#2[}}}

\newcommand{\kkdur}[2]{\ema{a_{#1,#2}}}
\newcommand{\kkaft}[1]{\ema{b_{#1}}}

\newcommand{\rwi}[2]{\ema{\mathit{ct}_{\mathrm{#1}}\left(#2\right)}}
\newcommand{\rw}[1]{\rwi{}{#1}}

\newcommand{\rwh}[1]{\ema{\xoverline[.88]{\mathit{ct}}\left(#1\right)}}

\newcommand{\ct}{\ema{P}}

\newcommand{\trl}{\ema{\ct_{\mathit{rl}}}}

\newcommand{\atrl}{\ema{\xoverline[.7]{\ct_{\mathit{rl}}}}}
\newcommand{\atps}{\ema{\xoverline[.7]{\ct_{\mathit{ps}}}}}

\newcommand{\tcom}{\ema{\ct_{\mathit{com}}}}

\newcommand{\fa}[1]{\ema{f_{#1}}}

\newcommand{\atrlf}{\ema{\trl^{(0)}}}

\newcommand{\atrli}{\ema{\trl^{(\mathrm{i})}}}

\newcommand{\scas}{\ema{\mathit{cc}}}
\newcommand{\fcas}{\ema{\mathit{cc}}}
\newcommand{\mem}{\ema{\mathit{rc}}}

\newcommand{\calrl}{\ema{\mathit{cw}}}

\newcommand{\rint}[4]{\ema{\int_{#1}^{#2} #3 \, \mathit{d#4}}}

\newcommand\vespv{\ema{s}}
\newcommand\vesp[1]{\ema{\vespv_{#1}}}

\newcommand{\avexp}[1]{\ema{\barov{e}\left(#1\right)}}
\newcommand{\reexp}[1]{\ema{e\left(#1\right)}}
\newcommand{\difavexp}[1]{\ema{\barov{e}'\left(#1\right)}}

\newcommand{\trlo}{\ema{\trl^{(0)}}}
\newcommand{\exppl}{(\texttt{+})}

\newcommand{\cw}{\ema{\mathit{cw}}}
\newcommand{\pw}{\ema{\mathit{pw}}}

\newcommand{\rc}{\ema{\mathit{rc}}}
\newcommand{\cc}{\ema{\mathit{cc}}}
\newcommand{\rlw}{\ema{\mathit{rlw}}}

\newcommand{\rlsiz}{\ema{\rlw^{\exppl}}}

\newcommand{\thru}{\ema{T}}

\newcommand{\ctot}{\ema{P}}
\newcommand{\thr}[1]{\ema{\mathcal{T}_{#1}}}

\newcommand{\shifti}[2]{\ema{\mathit{delay}_{#1}\left(#2\right)}}

\newcommand{\avsupe}[1]{\ema{\barov{\mathit{sp}}\left(#1\right)}}

\def\ffuf{\ema{\operatorname{f_1}}}
\def\ffus{\ema{\operatorname{f_2}}}
\def\ffusi{\ema{\ffus^{-1}}}

\newcommand{\fuf}[1]{\ema{\ffuf{\left( #1\right)}}}

\newcommand{\fus}[1]{\ema{\ffus\left( #1\right)}}
\newcommand{\fusi}[1]{\ema{\ffusi\left( #1\right)}}

\newcommand{\rl}{retry loop\xspace}
\newcommand{\rls}{retry loops\xspace}

\newcommand{\wl}{work loop\xspace}

\newcommand{\supw}{success period\xspace}
\newcommand{\Supw}{Success period\xspace}
\newcommand{\SUpw}{Success Period\xspace}
\newcommand{\supws}{success periods\xspace}

\newcommand{\cww}{critical work\xspace}

\newcommand{\pww}{parallel work\xspace}

\newcommand{\pwws}{parallel works\xspace}

\newcommand{\re}{retry\xspace}
\newcommand{\res}{retries\xspace}
\newcommand{\RES}{Retries\xspace}

\newcommand{\wc}{worst-case\xspace}

\newcommand{\avba}{average-based\xspace}
\newcommand{\Avba}{Average-based\xspace}

\newcommand{\ps}{parallel section\xspace}

\newcommand{\ds}{data structure\xspace}
\newcommand{\dss}{data structures\xspace}

\newcommand{\cas}{{\it CAS}\xspace}
\newcommand{\cass}{\cas{}'s\xspace}

\newcommand{\rf}{{\it Read}\xspace}

\newcommand{\acc}{{\it Access}\xspace}

\newcommand{\deq}{deque\xspace}
\newcommand{\deqs}{deques\xspace}
\newcommand{\Deq}{Deque\xspace}

\newcommand{\enqop}{\FuncSty{Enqueue}\xspace}
\newcommand{\deqop}{\FuncSty{Dequeue}\xspace}
\newcommand{\popop}{\FuncSty{Pop}\xspace}
\newcommand{\pushop}{\FuncSty{Push}\xspace}
\newcommand{\incop}{\FuncSty{Increment}\xspace}
\newcommand{\decop}{\FuncSty{Decrement}\xspace}

\newcommand{\ghz}[1]{\ema{#1\,\mathrm{GHz}}}
\newcommand{\megb}[1]{\ema{#1\,\mathrm{MB}}}

\newcommand{\cycles}[1]{\ema{#1\,\mathrm{cycles}}}
\newcommand{\uow}[1]{\ema{#1\,\mathrm{u.o.w.}}}

\newcommand{\aexpi}[1]{\ema{\barov{e_i}\left(#1\right)}}
\newcommand{\aexp}[2]{\ema{\barov{e_{#1}}\left(#2\right)}}

\newcommand{\expansion}[1]{\avexp{#1} }
\newcommand{\expansionp}[1]{\difavexp{#1} }
\newcommand{\stag}[1]{\ema{\mathit{Stage}_{#1}}}
\newcommand{\casn}[1]{\ema{\mathit{CAS}_{#1}}}

\newcommand{\pushl}{\FuncSty{PushLeft}\xspace}
\newcommand{\pushr}{\FuncSty{PushRight}\xspace}
\newcommand{\popl}{\FuncSty{PopLeft}\xspace}
\newcommand{\popr}{\FuncSty{PopRight}\xspace}

\newcommand{\anch}{{\it Anchor}\xspace}

\newcommand{\ie}{\textit{i.e.}\xspace}
\newcommand{\etal}{\textit{et al.}\xspace}
\newcommand{\eg}{\textit{e.g.}\xspace}

\newcommand{\afort}{\textit{a fortiori}\xspace}
\newcommand{\Afort}{\textit{A fortiori}\xspace}

\newcommand\rr[1]{#1}

\newcommand\pp[1]{}

\newcommand\pr[2]{#2}

\newcommand\tra[1]{}

\newcommand\falseparagraph[1]{\noindent{\bf #1:}\xspace}

\makeatletter
\newcommand{\removelatexerror}{\let\@latex@error\@gobble}
\makeatother

\newcommand{\abstalgo}{
\removelatexerror
\pp{\scriptsize}
\begin{procedure}[H]
\SetKwData{pet}{execution\_time}
\SetKwData{pdo}{done}
\SetKwData{psucc}{success}
\SetKwData{pcur}{current}
\SetKwData{pnew}{new}
\SetKwData{pacp}{AP}
\SetKwData{pttot}{t}
\SetKwFunction{pinit}{Initialization}
\SetKwFunction{ppw}{Parallel\_Work}
\SetKwFunction{pcw}{Critical\_Work}
\SetKwFunction{pread}{Read}
\SetKwFunction{pcas}{CAS}

\SetAlgoLined
\While{! \pdo}{\nllabel{alg:li-bwl}
\ppw{}\;\nllabel{alg:li-ps}
\While{! \psucc}{\nllabel{alg:li-bcs}
\pcur $\leftarrow$ \pread{\pacp}\;\nllabel{alg:li-bbcs}
\pnew $\leftarrow$ \pcw{\pcur}\;
\psucc $\leftarrow$ \pcas{\pacp, \pcur, \pnew}\;\nllabel{alg:li-ecs}
}
}
\caption{AbstractAlgorithm()\label{alg:gen-name}}
\end{procedure}
}

\usepackage{mathtools}

\newcommand{\Watiw}{Slack time\xspace}
\newcommand{\WAtiw}{Slack Time\xspace}
\newcommand{\watiw}{slack time\xspace}

\newcommand{\staw}{stage\xspace}

\newcommand{\staws}{stages\xspace}

\newcommand{\wati}[1]{\ema{\mathit{st}\left(#1\right)}\xspace}

\newcommand{\avwati}[1]{\ema{\xoverline[.8]{\mathit{st}}\left(#1\right)}\xspace}

\newcommand{\COmpw}{Completion time\xspace}
\newcommand{\compw}{completion time\xspace}

\DeclarePairedDelimiter\abs{\lvert}{\rvert}%
\DeclarePairedDelimiter\norm{\lVert}{\rVert}%

\makeatletter
\let\oldabs\abs
\def\abs{\@ifstar{\oldabs}{\oldabs*}}
\let\oldnorm\norm
\def\norm{\@ifstar{\oldnorm}{\oldnorm*}}
\makeatother

\pr{

\title{How Lock-free Data Structures Perform in Dynamic Environments:
 Models and Analyses}

\author[1]{Aras Atalar}
\author[2]{Paul Renaud-Goud}
\author[1]{Philippas Tsigas}

\affil[1]{Chalmers University of Technology, S-41296 G\"oteborg, Sweden\\
  \texttt{aaras|tsigas@chalmers.se}}
\affil[2]{Toulouse Institute of Computer Science Research, F-31062 Toulouse, France\\
  \texttt{prenaud@irit.fr}}

\authorrunning{A. Atalar, P. Renaud-Goud and P. Tsigas} 

\Copyright{Aras Atalar, Paul Renaud-Goud and Philippas Tsigas}


}{

\title{How Lock-free Data Structures Perform\\ in Dynamic Environments:\\ Models and Analyses}

\author{Aras Atalar\textsuperscript{\dag}, Paul Renaud-Goud\textsuperscript{$\star$} and Philippas Tsigas\textsuperscript{\dag}\\\textsuperscript{\dag} Chalmers University of Technology\\
  \textsuperscript{$\star$} Institut de Recherche en Informatique de Toulouse\\
  \vspace*{.3cm}
}

}

\begin{document}

\SetFuncSty{textsf}

\maketitle

\rr{\tableofcontents\newpage}

\newcommand{\donc}{\ema{\leadsto}}

\begin{abstract}
In this paper we present two analytical frameworks for calculating
the performance of lock-free \dss. Lock-free
\dss
are based on \rls and are called
by application-specific routines.
In contrast to previous work, we consider in this paper lock-free \dss
in dynamic environments. The size of each of the
\rls, and
the size of the
application routines invoked in between,
are not constant but may change dynamically. The new frameworks follow
two different approaches. The first framework, the simplest one, is
based on queuing theory. It introduces an \avba approach that
facilitates a
more coarse-grained analysis, with the benefit of being ignorant of
size distributions. Because of this independence from the distribution nature
it covers a set of complicated designs. The second approach,
instantiated with an exponential distribution for the size of the
application routines,
uses Markov chains, and is tighter because it
constructs stochastically the execution, step by step.

Both frameworks provide a performance estimate which is close to what
we observe in practice. We have validated our analysis on (i) several
fundamental lock-free \dss such as stacks, queues, \deqs and counters,
some of them employing helping mechanisms, and (ii) synthetic
tests covering a wide range of possible lock-free designs.
We show the applicability of our results by introducing new back-off
mechanisms, tested in application contexts, and by designing an
efficient memory management scheme that typical lock-free algorithms
can utilize.

\end{abstract}

\rr{\clearpage}

\newcommand{\ra}{\ema{\rightarrow}}



\newcommand\lemsl{
\begin{lemma}
\label{lem:unif-min}
Let an integer $n$, a real positive number $a$, and $n$
independent random variables $X_1, X_2, \dots, X_n$, uniformly
distributed within $[0,a[$. Let then $X$ be the random variable
defined by: $X = \min_{i \in \inte[1]{n}} X_i$. The expectation of $X$ is:
\[ \expe{X} = \frac{a}{n+1}. \]
\end{lemma}
\begin{proof}
Let a positive real number $x$ be such that $x<a$. We have
\begin{align*}
\pro{X > x} &=  \pro{\forall i : X_i > x}\\
&= \prod_{i=1}^n \pro{X_i > x}\\
\pro{X > x} &=  \left( \frac{a-x}{a} \right) ^{n}\\
\end{align*}
Therefore, the probability distribution of $X$ is given by:
\[ t \mapsto \frac{n}{a} \left( \frac{a-x}{a} \right) ^{n-1},\]
and its expectation is computed through

\begin{align*}
\expe{X} &= \frac{n}{a} \rint{0}{a}{x \times \left( \frac{a-x}{a} \right) ^{n-1}}{x} \\
&= \frac{n}{a} \rint{0}{a}{(a-u) \times \left( \frac{u}{a} \right) ^{n-1}}{u} \\
&= \frac{n}{a^n} \rint{0}{a}{(a-u) \times u^{n-1} }{u} \\
&= \frac{n}{a^n} \left( a \times \frac{a^n}{n} - \frac{a^{n+1}}{n+1} \right) \\
\expe{X} &= \frac{a}{n+1}.
\end{align*}
\end{proof}
}

\newcommand\proofswitch{
\pr{We have}{It remains} to decide whenever the \ds is under contention or not, and
to find the corresponding solution.
Concerning the frontier between contended and non-contended system, we
can remark that Equations~\ref{eq:little-co} and~\ref{eq:little-nc}
are equivalent if and only if
\begin{equation}
\label{eq:little-fronti}
\frac{\rc + \cw + \scas}{\atrl} =
\frac{\atrl +2}{\atrl +1}  \left( \cw + \avexp{\atrl} \right)
+ 2\scas,
\end{equation}
which leads to Lemma~\ref{lem:lit-swi}.

\begin{lemma}
\label{lem:lit-swi}
The system switches from being
non-contended to being contended at $\atrl = \atrlf$, where
\pr{
\[ \atrlf = \frac{\scas+\cw-\mem}{2 (\cw + 2 \scas)} \left( \sqrt{1+\frac{4 (\mem+\cw+\scas) (\cw+2\scas)}{(\scas+\cw-\mem)^2}} -1 \right). \]
}
{
\[ \atrlf = \frac{-(\scas+\cw-\mem) + \sqrt{\left( \scas+\cw-\mem \right)^2 + 4 (\mem+\cw+\scas) (\cw+2\scas)}}{2 (\cw + 2 \scas)}.\]
}
\end{lemma}
\begin{proof}
We show that:
\begin{itemize}
\item \atrlf is the unique positive solution of Equation~\ref{eq:little-fronti} if the expansion is set to 0,
\item $\atrlf \leq 1$,
\item there is no solution of Equation~\ref{eq:little-fronti} with a non-null expansion.
\end{itemize}

If the expansion is set to 0, then Equation~\ref{eq:little-fronti} can
be turned into the second order equation
\[ \atrl^2 (\cw + 2 \cc) + \atrl \left( \cw + \cc - \mem \right) - (\mem + \cw + \cc) = 0,
\]
that has a single positive solution: \atrlf.

While instantiating the binomial with $\atrl = 1$, we obtain $\cw + 2
(\scas - \mem)$, which is not negative, since $\scas \geq \mem$ in all
the architectures that we are aware of.
As the second order equation has also a negative solution, and $\cw +
2 \cc$ is positive, we have that $1 \geq \atrlf$.
This implies that \atrlf is a solution of the former
Equation~\ref{eq:little-fronti}: the expansion is indeed a
non-decreasing function, thus $0 \leq \avexp{\atrlf} \leq \avexp{1} =
0$. Still we could have other solutions with a non-null expansion.

However, Equation~\ref{eq:little-fronti} can be rewritten as:
\begin{equation}
\label{eq:lit-fro-mon}
\rc + \cw + \scas =
\frac{\atrl +2}{\atrl +1} \times \atrl \times  \left( \cw + \avexp{\atrl} \right)
+ 2\scas.
\end{equation}
The left-hand side of Equation~\ref{eq:lit-fro-mon} is
constant, while the right-hand side is increasing, which discards
any other solution, hence the lemma.
\end{proof}

Thanks to Lemma~\ref{lem:lit-swi}, we can unify the \supw as:
\[
\avsupe{\atrl} = \left\{\begin{array}{ll}
\left( \mem + \cw + \scas \right) / \atrl & \quad \text{if } \atrl \leq \atrlf \\
 \left( \cw + \avexp{\atrl}\right) \times \frac{\atrl + 2}{\atrl +1} + 2\scas  &
 \quad \text{otherwise.}
\end{array}\right.
\]
The unified \supw obeys to the following equation
\begin{equation}
\label{eq:little-res}
\avsupe{\atrl} = \frac{\pw}{\ct - \atrl}.
\end{equation}
}

\newcommand{\prooffp}{
Let us note $\fuf{\atrl} = \avsupe{\atrl} \times \atrl$ and $\fus{\atrl} =
\pw \times \atrl / (\ct-\atrl)$; then Equation~\ref{eq:little-res} is equivalent to
$\fuf{\atrl} = \fus{\atrl}$, and we have some properties on \ffuf and \ffus.

Firstly, since $x \mapsto x(x+2)/(x+1)$ is non-decreasing on
$[0,\pinf[$, as well as the expected expansion, we know that \ffuf is
a non-decreasing function.
Secondly, \ffus is increasing on $[0,\ct[$, and is bijective from
$[0,\ct[$ to $[0,\pinf[$. We can thus rewrite Equation~\ref{eq:little-res} as:
\begin{equation}
\label{eq:fi-po}
 \atrl = \fusi{\fuf{\atrl}}.
\end{equation}
Moreover, $\ffusi \circ \ffuf$ is a non-decreasing function, as a
composition of two non-decreasing functions.
Thirdly, \ffusi can be obtained through $x = \fus{\fusi{x}} = \pw \times \fusi{x} / (\ct -
\fusi{x})$, which leads to
\[ \fusi{x} = \frac{x}{\pw + x} \ct. \]

In addition, we know by construction that if $\atrl > \atrlf$, then
\begin{equation}
\label{eq:little-fip}
 \left( \cw + \avexp{\atrl}\right) \times \frac{\atrl + 2}{\atrl +1} + 2\scas \geq \frac{\mem + \cw + \scas}{\atrl}.
\end{equation}
Indeed, on the one hand,
\[ \lim_{\atrl \rightarrow 0^{+}} \frac{\mem + \cw + \scas}{\atrl} = \pinf, \]
and on the other hand, $(\cw + \avexp{\atrl}) \times (\atrl + 2)/(\atrl +1) + 2\cc$
remains bounded. According to Lemma~\ref{lem:lit-swi}, those two
functions cross only once, hence Equation~\ref{eq:little-fip}.

Since $\avsupe{\atrl} = (\mem + \cw + \scas)/\atrl$ if $\atrl \leq
\atrlf$, we have $\avsupe{\atrl} \geq (\mem + \cw + \scas)/\atrl$
for any \atrl, and then
\[ \fuf{\atrl} \geq \mem + \cw + \scas. \]

Let then
\[ \atrli = \frac{\mem + \cw + \scas}{\pw + \mem + \cw + \scas} \ct.\]
We have seen that $\ffusi \circ \ffuf$ is a non-decreasing function,
hence
\begin{align*}
\fusi{\fuf{\atrli}} &\geq \fusi{\mem + \cw + \scas}\\
&\geq \frac{\mem + \cw + \scas}{\pw +  \mem + \cw + \scas} \times \ct\\
\fusi{\fuf{\atrli}} &\geq \atrli.
\end{align*}
Since \ffusi is bounded, Equation~\ref{eq:fi-po} admits a solution.

We are interested in the solution whose \atrl is minimal since it
corresponds to the first attained solution when the expansion grows,
starting from 0. The current theorem comes then from the application
of the Knaster-Tarski theorem.
}


\newcommand\proofaliexp{
Let us set the timeline so that at the beginning of the \supw, \ie
just after a successful \cas, we are at $t=0$.
Firstly, a success cannot start before $t=t_0$, where
$t_0=\scas+\lceil\cw/\scas\rceil\scas$. The quickest thread indeed
starts a failed \cas at $t=0$ and comes back from \cww at
$t=\scas+\cw$. It has then to wait for the current \cas to finish
before being able to obtain the cache line.
At $t=t_0$, $\trl - t_0/\scas +1$ threads are competing for the
data. Among them, 1 thread will lead to a successful \cas, while the
$\trl - t_0/\scas$ other threads will end up with a failed \cas.
If a failed \cas occurs, then at $t=t_0+\scas$, the same number of
threads compete, but now there is one more potential success and one
less potential failure.  In the worst case, it will continue until all
competing threads will lead to a successful \cas.

Let $\tcom = \trl - t_0/\scas +1$ the number of threads that are
competing at each round, and let, for all $i \in \inte[1]{\tcom}$,
$p_i = i/\tcom$ the probability to draw a thread that will execute a
successful \cas.

The expected number of failed \cass that occurs after the first thread comes back is then given by
\[
\expe{F} = \prsu{1} \times 0 + (1-\prsu{1}) \prsu{2} \times 1 + \dots + \\
(1-\prsu{1})(1-\prsu{2})\times\dots\times (1-\prsu{\tcom -1}) \times \prsu{\tcom} \times (\tcom -1).
\]
More formally,
\begin{align*}
\expe{F} &= \sum_{i=1}^{\tcom} \prod_{j=1}^{i-1} (1-\prsu{j}) \prsu{i} \times (i-1)\\
&= \sum_{i=1}^{\tcom} \prod_{j=1}^{i-1} (1-\frac{j}{\tcom}) \frac{i}{\tcom} \times (i-1)\\
&= \sum_{i=1}^{\tcom} \frac{1}{\left(\tcom\right)^i} \prod_{j=1}^{i-1} (\tcom-j) i(i-1)\\
\expe{F} &= \sum_{i=1}^{\tcom} \frac{i(i-1)}{\left(\tcom\right)^i} \frac{(\tcom-1)!}{(\tcom-i)!}\\
\end{align*}
}


\def\mmtr{M}
\def\mmab{C}
\def\mma{D}
\def\mmpe{Q}
\def\mmtrir{R}
\def\mmtril{T}

\newcommand{\mtr}[2]{\ema{\mmtr_{#1,#2}}}
\newcommand{\mab}[2]{\ema{\mmab_{#1,#2}}}
\newcommand{\ma}[2]{\ema{\mma_{#1,#2}}}
\newcommand{\mpe}[2]{\ema{\mmpe_{#1,#2}}}
\newcommand{\mtrir}[2]{\ema{\mmtrir_{#1,#2}}}
\newcommand{\mtril}[2]{\ema{\mmtril_{#1,#2}}}

\newcommand{\eigv}{\ema{v}}
\newcommand{\eig}[1]{\ema{\eigv_{#1}}}

\newcommand\transmat{
We consider here that the system is in a given state, and\rr{ we} compute
the probability that the system will next reach any other
state. Without loss of generality, we\rr{ can} choose the origin of time
such that the current \supw begins at $t=0$.

Let us first look at the core cases, \ie let $i \in \intmid \cup \inthig$ and
$k \in \inte[0]{\ct-i-1}$; we assume that the system is currently in
state \sta{i}, and we are interested in the probability that the
system will switch to \sta{i+k} at the end of the current state. In
other words, we want to find the probability that, given that the
current \supw started when $i$ threads were in the \rl, the next
\supw will begin while $i+k$ threads are in the \rl.

\rr{
\begin{figure}
\begin{center}
\begin{tikzpicture}[%
it/.style={%
    rectangle,
    text width=11em,
    text centered,
    minimum height=\pr{2}{3}em,
    draw=black!50,
    scale=.9,
  }
]
\coordinate (O) at (0,0);
\pha{pcas}{O}{\cas}{\pr{2}{5}}{\green}
\pha{sla}{pcas}{\wati{i}}{\pr{5}{8}}{\grey}
\pha{acc}{sla}{\cas}{\pr{2}{4}}{red}
\pha{cri}{acc}{\cw}{\pr{3}{4}}{\maroon}
\pha{exp}{cri}{\reexp{i}}{\pr{4}{6}}{\grey}
\pha{fcas}{exp}{\cas}{\pr{2}{5}}{\green}
\coordinate (intsl) at ($(sla.north west)+(0,1*\marup)$); 
\coordinate (intsr) at ($(sla.north east)+(0,1*\marup)$);
\coordinate (inte) at ($(fcas.north east)+(0,1*\marup)$); 
\draw [decorate,decoration={brace,amplitude=10pt}]
(intsl)  --  (intsr)
node [black,midway,yshift=15pt] {\scriptsize 0 new thread};

\draw [decorate,decoration={brace,amplitude=10pt}]
(intsr) -- (inte)
node [black,midway,yshift=15pt] {\scriptsize $k+1$ new threads};
\arcod{($(acc.north west)!.2! (acc.north east) + (0,.7*\marup)$)}{.4*\marup}
\arcod{($(exp.north west)!.1! (exp.north east) + (0,.7*\marup)$)}{.4*\marup}
\arcod{($(exp.north west)!.9! (exp.north east) + (0,.7*\marup)$)}{.4*\marup}

\coordinate (extsl) at ($(sla.south west)+(0,-\marup)$);
\coordinate (extsr) at ($(sla.south east)+(0,-\marup)$);
\draw [decorate,decoration={brace,mirror,amplitude=10pt},text width=11em, align=center] (extsl) -- (extsr)
node (caca) [black,midway,yshift=-18pt,execute at begin node=\setlength{\baselineskip}{8pt}] {\scriptsize at least 1\\new thread};

\coordinate (Ob) at ($(sla.south west)!.2!(sla.south) + (0,-3*\marup)$);
\pha{accb}{Ob}{\rf}{\pr{2}{4}}{yellow}
\pha{crib}{accb}{\cw}{\pr{3}{4}}{\maroon}
\pha{expb}{crib}{\reexp{i}}{\pr{4}{6}}{\grey}
\pha{fcasb}{expb}{\cas}{\pr{2}{5}}{\green}
\arcou{(accb.south west)}{3.3*\marup}
\arcou{($(crib.south west)!.2! (crib.south east) + (0,-.7*\marup)$)}{.4*\marup}
\arcou{($(expb.south west)!.9! (expb.south east) + (0,-.7*\marup)$)}{.4*\marup}
\draw[very thick, draw=blue] (accb.south west) -- (fcasb.south east) -- (fcasb.north east) -- (accb.north west);
\coordinate (extnl) at ($(accb.south west)+(0,-\marup)$); 
\coordinate (extnr) at ($(fcasb.south east)+(0,-\marup)$);
\draw [decorate,decoration={brace,mirror,amplitude=10pt},text width=11em, align=center]
(extnl) -- (extnr)
node [black,midway,yshift=-15pt] {\scriptsize $k$ new threads};

\draw[dotted, draw=black] (intsl) -- (extsl);
\draw[dotted, draw=black] (intsr) -- (extsr);
\draw[dotted, draw=black] (inte) -- (fcas.south east);
\draw[dotted, draw=black] (accb.north west) -- (extnl);
\draw[dotted, draw=black] (fcasb.north east) -- (extnr);
\node[text width=1.2cm, align=center] (inttext) at ($(pcas.west) + (-2*\marup,0)$) {Internal\\ execution};
\node[anchor=east] (eint) at ($(inttext.east) + (0.2,1.2*\marup)$) {\eve{int}};
\node[anchor=east] (eext) at ($(inttext.east) + (0.2,-2*\marup)$){\eve{ext}};

\path[->,out=90,in=-180] ($(inttext.north west)!.3!(inttext.north)$) edge (eint);
\path[->,out=-90,in=-180] ($(inttext.south west)!.3!(inttext.south)$) edge (eext);
\end{tikzpicture}
\end{center}
\caption{Possible executions\label{fig:ex-eint-eext}}
\end{figure}
}

As the successful thread will exit the \rl at the end of the current
\supw, there is at least one thread that enters the \rl during the
current \supw. Two non-overlapping events can then occur (see
Figure~\ref{fig:ex-eint-eext}): either the first thread exiting the
\ps starts within $[0,\wati{i}[$, \ie in the \watiw of the internal
    execution, and this event is written \eve{ext}, or the first
    thread entering the \rl starts after $t=\wati{i}$, and this event
    is denoted by \eve{int}.
Therefore, we have $\pro{\sta{i} \rightarrow \sta{i+k}} =
\pro{\eve{ext}} + \pro{\eve{int}}$.

First note that \eve{ext} cannot happen when the \supw is highly
contended; in this case, the \watiw is indeed null, and we conclude
$\pro{\eve{ext}} = 0$. In addition, we have seen in
Section~\ref{sec:mark-expa} that external threads, \ie threads that
are in the \ps at the beginning of the \supw, do not participate to
the game of expansion, so they cannot be successful. Under
high-contention, \eve{int} happens, and the successful \cas{} that
ends the \supw is operated by an internal thread, \ie a thread that
was already in the \rl when the \supw began.

Under medium contention, \eve{ext} can occur. In this case, an
external thread accesses the \ds before any internal thread does. We
have also seen that the expansion is null in medium contention level,
thus the external thread will execute its \cww, and especially its
\cas without being delayed; this implies that the first external
thread that accesses the \ds will end the current \supw with the end
of its \cas. If however \eve{int} occurs, an internal thread succeeds,
but is not necessarily the first thread that accessed the \ds during
the \supw.

The two possible events are pictured in Figure~\ref{fig:ex-eint-eext},
where the blue arrows represent the threads that exit the \ps. Recall,
we aim at computing the probability to start the next \supw with $i+k$
threads inside the \rl. We formalize the idea drawn in the figure by
using \thex{a}{b}, which is defined as a random variable
indicating the number of threads exiting the \ps during the time
interval $[a,b[$.
The probability of having \eve{int} is then given by\vspp{0}
\pr{\begin{equation*}
\pro{\eve{int}} = \pro{\thex{0}{\wati{i}}=0 \quad | \quad \trl=i \text{ at } t=0^{+}}
\times \pro{\thex{\wati{i}}{\wati{i}+\rw{i}}=k+1 \quad | \quad \trl=i \text{ at } t=\wati{i}^{+}}.
\end{equation*}}
{

\begin{align*}
\pro{\eve{int}} =& \pro{\thex{0}{\wati{i}}=0 \quad | \quad \trl=i \text{ at } t=0^{+}} \\
&\times \pro{\thex{\wati{i}}{\wati{i}+\rw{i}}=k+1 \quad | \quad \trl=i \text{ at } t=\wati{i}^{+}}
\end{align*}

}

Concerning \eve{ext}, we know that if $i \in \inthig$, then $\pro{\eve{ext}} = 0$.
Otherwise, if we denote by $t_3$ the starting time of the first
thread that exits the \ps, we obtain
\begin{align*}
\pro{\eve{ext}} =& \pro{\thex{0}{\wati{i}}>0 \quad | \quad \trl=i \text{ at } t=0^{+}} \\
&\times \pro{\thex{t_3}{t_3+\rc+\cw+\scas}=k \quad | \quad \trl=i+1 \text{ at } t=t_3^{+}}
\end{align*}
To simplify the reasoning, and given that the costs of \rf and \cas
are approximately the same, we approximate $t_3+\rc+\cw+\scas$ with $t_3+\scas+\cw+\scas$, leading to
\begin{align*}
\pro{\eve{ext}} =& \pro{\thex{0}{\wati{i}}>0 \quad | \quad \trl=i \text{ at } t=0^{+}} \\
&\times \pro{\thex{t_3}{t_3+\rw{i+1}}=k \quad | \quad \trl=i+1 \text{ at } t=t_3^{+}}
\end{align*}

According to the exponential distribution, given a thread that is in
the \ps at $t=a$, the probability to exit the \ps within $[a,b[$ is:
\[ \rint{a}{b}{\lambda \expu{-\lambda (t-a)}}{t} = \rint{0}{b-a}{\lambda \expu{-\lambda u}}{u}.  \]
It is also the probability, given a thread that is in the \ps at
$t=0$, to exit the \rl within $[a,b-a[$. This implies:
\begin{align*}
\pro{\eve{int}} =& \pro{\thex{0}{\wati{i}}=0 \quad | \quad \trl=i \text{ at } t=0^{+}} \\
&\times \pro{\thex{0}{\rw{i}}=k+1 \quad | \quad \trl=i \text{ at } t=0^{+}}
\end{align*}
and
\begin{align*}
\pro{\eve{ext}} =& \pro{\thex{0}{\wati{i}}>0 \quad | \quad \trl=i \text{ at } t=0^{+}} \\
&\times \pro{\thex{0}{\rw{i}}=k \quad | \quad \trl=i+1 \text{ at } t=0^{+}}.
\end{align*}

To lighten the notations, let us define
\begin{equation}
\label{eq:def-ab}
\left\{ \begin{array}{l}
\kkdur{i}{k} = \pro{\thex{0}{\rw{i}}=k \quad | \quad \trl=i \text{ at } t=0}\\
\kkaft{i} = \pro{\thex{0}{\wati{i}} = 0 \quad | \quad \trl=i \text{ at } t=\rw{i}^{+}}.
\end{array} \right.
\end{equation}

In addition, given a thread that is in the \ps at $t=0$, the
probability to exit the \ps within $[0,b-a[$ is $\rint{0}{b-a}{\lambda
      \expu{-\lambda u}}{u}$. By counting the number of threads that
    need to exit the \ps, we obtain:
\begin{equation}
\label{eq:exp-ab}
\left\{ \begin{array}{l}
\kkdur{i}{k} = \binom{\ct-i}{k} \left( 1 - \expu{-\lambda \rw{i}} \right) ^{k} \left( \expu{-\lambda \rw{i}} \right) ^{\ct-i-k}\\
\kkaft{i} = \left( \expi{-\lambda \wati{i}} \right) ^{\ct-i}.
\end{array} \right.
\end{equation}

Altogether, we have that
\[ \pro{\sta{i} \rightarrow \sta{i+k}} = \kkaft{i} \times \kkdur{i}{k+1} + (1-\kkaft{i}) \times \kkdur{i+1}{k}. \]

\medskip

The situation is slightly different if $k=-1$; in this case, no thread
should exit the \ps during the \watiw and no thread should exit during
the \re of the first thread that accessed the \ds during the
\supw neither. This shows that
\[ \pro{\sta{i} \rightarrow \sta{i-1}} = \kkaft{i} \times \kkdur{i}{0}. \]

When the \supw is not contended, \ie if $i=0$, the \watiw of the
execution that ignores external threads can be seen as infinite, hence
we can define $\kkaft{0} = 0$ (the probability that a thread exits its
\ps during an infinite interval of time is $1$). As for the
\kkdur{i}{k}'s, they can be defined in the same way as earlier.

We have obtained the full transition matrix $\left( \mtr{i}{j}
\right)_{(i,j) \in \inte{\ct-1}^2}$, which is a triangular matrix,
augmented with a subdiagonal:
\begin{equation*}
\left\{
\begin{array}{lll}
\mtr{i}{i+k} &=  \kkaft{i} \kkdur{i}{k+1} + (1-\kkaft{i}) \kkdur{i+1}{k} & \text{ if }
k \in \inte[0]{\ct-i-1}\\
\mtr{i}{i-1} &= \kkaft{i} \times \kkdur{i}{0} & \text{ if }
i > 0\\
\mtr{i}{j} &= 0 & \text{ otherwise}\\
\end{array}
\right.
\end{equation*}

\medskip

\begin{lemma}
$\mmtr$ is a right stochastic matrix.
\end{lemma}
\begin{proof}
First note that, by definition of \kkdur{i}{k}, for all $i \in \inte[0]{\ct-1}$,
\[ \sum_{k=0}^{\ct-i} \kkdur{i}{k} = 1. \]
If $i$ threads are indeed inside the \rl at $t=0$, then, within
$[0,\wati{i}[$, at least $0$ thread, and at most $\ct-i$ threads (inclusive) will exit
their \ps.

We have first
\[
\sum_{j=0}^{\ct-1} \mtr{0}{j} = \sum_{k=0}^{\ct-1} \kkdur{0+1}{k} = 1.
\]

In the same way, for all $i \in \inte[1]{\ct-1}$,
\begin{align*}
\sum_{j=0}^{\ct-1} \mtr{i}{j} =&\; \sum_{k=-1}^{\ct-1-i} \mtr{i}{i+k}\\
=&\; \kkaft{i} \times \kkdur{i}{0}  + \sum_{k=0}^{\ct-1-i} \kkaft{i} \kkdur{i}{k+1} + (1-\kkaft{i}) \kkdur{i+1}{k}\\
=&\; \kkaft{i} \times \sum_{k=-1}^{\ct-1-i} \kkdur{i}{k+1}  +   (1-\kkaft{i}) \sum_{k=0}^{\ct-1-i}  \kkdur{i+1}{k}\\
\sum_{j=0}^{\ct-1} \mtr{i}{j} =&\;1.
\end{align*}

\end{proof}

\begin{lemma}
The transition matrix has a unique stationary distribution, which is
the unique left eigenvector of the transition matrix with eigenvalue 1
and sum of its elements equal to 1.
\end{lemma}
\begin{proof}
Note that the Markov chain is
irreducible and aperiodic. Let $X \geq \ct -1$, $i \in \inte{\ct-1}$ and
$j \in \inte[i]{\ct-1}$.
\begin{align*}
\pro{\sta{j} \rightarrow \sta{i} \text{ in X steps}}  \geq &
\pro{\sta{j} \rightarrow \sta{j-1} \rightarrow \dots \rightarrow \sta{i}}\\
& \times \pro{\sta{i} \rightarrow \sta{i}}^{X-(j-i)}\\
\pro{\sta{j} \rightarrow \sta{i} \text{ in X steps}} > & 0
\end{align*}
As
\[ \pro{\sta{i} \rightarrow \sta{j} \text{ in X steps}} \geq \pro{\sta{i} \rightarrow \sta{j}} > 0, \]
the Markov chain is irreducible.
Since \sta{1} is clearly aperiodic, and the chain is irreducible, the chain is aperiodic as well.

This implies that the Markov chain has a unique stationary
distribution, which is the unique left eigenvector of the transition
matrix with eigenvalue 1 and sum of its elements equal to 1.
\end{proof}
}


\newcommand{\statdis}{
\begin{theorem}
Given the transition matrix, the stationary distribution can be found
in $(\ct+1)\ct -1$ operations.
\end{theorem}
\begin{proof}
As the Markov chain is irreducible, the stationary
distribution does not contend any zero. The space of the left
eigenvectors with unit eigenvalue is uni-dimensional; therefore, for
any \eig{0}, there exists a vector $\eigv = (\eig{0} \; \eig{1} \; \dots \;
\eig{\ct-1})$, such that \eigv spans this space.

Let $\eig{0}$ a real number; necessarily, \eigv fulfills $\eigv \cdot
\mmtr = \eigv$, hence for all $i \in \inte{\ct-2}$
\[ \sum_{k=0}^{i+1} \eig{k} \mtr{k}{i} = \eig{i}, \]
which leads to, for all $i \in \inte{\ct-2}$:
\[ \eig{i+1} = \frac{1}{\mtr{i+1}{i}} \left( (1-\mtr{i}{i})\eig{i} - \sum_{k=0}^{i-1} \eig{k} \mtr{k}{i} \right). \]
So we obtain the $\eig{1}, \dots, \eig{\ct-1}$ iteratively (we know
that $\mtr{i+1}{i} = \kkaft{i+1} \times \kkdur{i+1}{0}$, which is not
null), with $2\times i + 1$ operations needed to compute \eig{i+1}.

The elements of the stationary distribution should sum to one, so we
start from any \eig{0}, compute the whole vector, and then normalize
each element by their sum, hence the theorem.
\end{proof}

}


\newcommand{\watithput}{
In order to compute the final throughput, we have to compute the
expectation of the \watiw, when the system goes from state \sta{i} to
any other state, that we note \expe{\wati{\sta{i} \rightarrow
    \sta{\star}}}.
Also, we will be able to exhibit a vector $\vespv = (\vesp{0} ,
\vesp{1} , \dots , \vesp{\ct-1})$ of expected \supw, where \vesp{i} is
the expectation of the execution time of the \supw if $i$ threads are
in the \rl when the \supw begins:
\begin{equation*}
\left\{\begin{array}{ll}
\vesp{i} = \expe{\wati{\sta{i} \rightarrow \sta{\star}}} + \scas + \cw + \reexp{i} + \scas&\;\text{if }
i \notin \intnoc\\
\vesp{i} = \expe{\wati{\sta{i} \rightarrow \sta{\star}}} + \mem + \cw + \scas&\;\text{otherwise.}
\end{array}\right.
\end{equation*}

Finally, the expected throughput (inverse of the \supw) is calculated through
\pr{$\thru = (\eigv \cdot \vespv)^{-1},$}
{\[ \thru = \frac{1}{\eigv \cdot \vespv},\]}
where \eigv is the stationary distribution of the Markov chain.

 We know already that if
  $i \in \inthig$, then $\expe{\wati{\sta{i} \rightarrow \sta{i+k}}} =
  0$.

In the other extreme case, \ie if $i \in \intnoc$, we rely on the
following lemma.

\begin{lemma}
\label{lem:exp-min}
Let an integer $n$, a real number $\lambda$, and $n$ independent random
variables $X_1, X_2, \dots, X_n$, following an exponential
distribution of mean $\lambda^{-1}$. Let then $X$ be the random variable
defined by: $X = \min_{i \in \inte[1]{n}} X_i$. The expectation of $X$
is:
\[ \expe{X} = \frac{1}{\lambda n}. \]
\end{lemma}
\begin{proof}
We have
\begin{align*}
\pro{X > x} &=  \pro{\forall i : X_i > x}\\
&= \prod_{i=1}^n \pro{X_i > x}\\
&=  \left( \int_{x}^{\pinf}  \lambda \expu{-\lambda t} \right) ^{n}\\
\pro{X > x} &=  \expu{-\lambda n x}\\
\end{align*}
Therefore, the probability distribution of $X$ is given by:
\[ t \mapsto \lambda n \expu{-\lambda n t},\]
and its expectation is computed through

\begin{align*}
\expe{X} &= \rint{0}{\pinf}{ \lambda n t \expu{-\lambda n t}}{t}\\
&= \left[ \expu{-\lambda n t} t \right]_{\pinf}^0  + \rint{0}{\pinf}{\expu{-\lambda n t}}{t}\\
&= \left[ \frac{1}{\lambda n} \expu{-\lambda n t} \right]_{\pinf}^0 \\
\expe{X} &= \frac{1}{\lambda n}\\
\end{align*}
\end{proof}

This proves that
\[ \expe{\wati{\sta{0} \rightarrow \sta{\star}}} = \frac{\pw}{\ct}. \]

Let now $i \in \intmid$, and $k \in \inte[-1]{\ct -i -1}$; we are
interested in \expe{\wati{\sta{i} \rightarrow \sta{i+k}}}.
The \watiw is less immediate, and we use the following reasoning.
First note that the probability distribution of the first thread
exiting the \ps is given by $t \mapsto \lambda (\ct-i) \expu{-\lambda (\ct-i)
  t}$. If this thread comes back during $]0,\wati{i}[$, the time that
    passed since the beginning of the \supw is the \watiw, otherwise,
    it is \wati{i}.

\begin{align*}
\expe{\wati{\sta{i} \rightarrow \sta{\star}}}
&= \rint{0}{\wati{i}}{\lambda (\ct-i) \expu{-\lambda (\ct-i) t} t}{t}
+ \rint{\wati{i}}{\pinf}{\lambda (\ct-i) \expu{-\lambda (\ct-i) t} \wati{i}}{t}\\
&= \left[ \expu{-\lambda (\ct-i) t} t \right]_{\wati{i}}^{0}
+ \left[ \frac{1}{\lambda (\ct-i)} \expu{-\lambda (\ct-i) t} \right]_{\wati{i}}^{0}
+ \wati{i} \left[ \expu{-\lambda (\ct-i) t} \right]_{\pinf}^{\wati{i}} \\
\expe{\wati{\sta{i} \rightarrow \sta{\star}}}&= - \wati{i}\expu{-\lambda (\ct-i) \wati{i}}
+ \frac{1- \expu{-\lambda (\ct-i) \wati{i}}}{\lambda (\ct-i)}
+ \wati{i} \left( \expu{-\lambda (\ct-i) \wati{i}} \right)
\end{align*}

We conclude that
\[ \expe{\wati{\sta{i} \rightarrow \sta{\star}}} = \frac{1 - \expu{-\frac{(\ct-i) \wati{i}}{\pw}}}{\ct-i} \pw. \]

Putting all together, we obtain
\begin{equation*}
\left\{\begin{array}{ll}
\expe{\wati{\sta{i} \rightarrow \sta{\star}}} = \frac{1 - \expu{-\frac{(\ct-i) \wati{i}}{\pw}}}{\ct-i} \pw
& \quad\text{if } i \in \intnoc \cup \intmid\\
\expe{\wati{\sta{i} \rightarrow \sta{\star}}} = 0 & \quad \text{if } i \in \inthig.
\end{array}\right.
\end{equation*}
}


\newcommand\failedres{
Another metric to estimate the quality of the model is the number of
failed \res per successful \re. We compute it by counting the number
of failed \res within the current \supw, where a \re is billed to a
given \supw if its failed \cas occurs during this \supw. We denote by
$\expe{\fa{i}}$ the expected number of failed \cas during a \supw that
begins with $i$ threads, where $i \in \inte{\ct-1}$.

If the \supw is not contended, \ie if $i \in \intnoc$, no failure will
occur since the first \cas of the \supw will be a success; hence
$\expe{\fa{i}} = 0 = i$.

If the \supw is mid-contended, \ie if $i \in \intmid$, every thread
that is in the \rl in the beginning of the \supw will execute at least
one \cas during this \supw, and exactly two if the thread is the
successful one. We know indeed that, even if a thread exits its \ps
during the \watiw, and is then successful, the failed \cass will occur
before the thread entering the \rl executes its successful \cas. As
any thread that exits its \ps during the \supw either is successful at
its first \cas, or does not operate the \cas during the \supw, we
conclude that: $\expe{\fa{i}} = i$.

If the \supw is highly contended, \ie if $i \in \inthig$, then we know
that we have an uninterrupted sequence of failed \cass, from the
beginning of the \supw to the last ending successful \cas. The
expected number of failed \cass is then directly related to the
expected duration of the \supw. Recalling that the expansion is given
in Theorem~\ref{th:mark-expa}, we obtain:
\[ \expe{\fa{i}} = 1 + \frac{\cw + \reexp{i}}{\scas}. \]

}


\newcommand{\wholemarkov}{
\subsubsection{Transition Matrix}

\transmat

\subsubsection{Stationary Distribution}

\statdis

\subsubsection{\Watiw and Throughput}

\watithput

\subsubsection{Number of Failed \RES}
\label{sec:nbf}

\failedres
}


\newcommand\treib{
The lock-free stack by Treiber~\cite{lf-stack} is a fundamental \ds
that provides \popop and \pushop operations. To \popop an element, the
top pointer is read and the next pointer of the initial element is
obtained. The latter pointer will be the new value of the \cas that
linearizes the operation. So, accessing the next pointer of the
topmost element represents \cw as it takes place between the \rf and
the \cas.
We initialize the stack by pushing elements with or without
a stride from a contiguous chunk of memory. By this way, we are able
to introduce both costly or not costly cache misses. We also vary the
number of elements popped at the same time to obtain different \cw;
the results, with different \cw values are illustrated
in Figure~\ref{fig:stack}.}

\newcommand\synth{
We first evaluate our models using a set of synthetic tests
that have been constructed to abstract
different possible design patterns of lock-free data structures (value of \cw)
and different application contexts (value of \pw).
The \cww is either constant, or follows a Poisson distribution; in
Figure~\ref{fig:synt}, its mean value \cw is indicated at the top of
the graphs.

A steep decrease in throughput, as \pw gets low, can be observed for the cases with
low \cw, that mainly originates due to expansion.
When \cw is high, performance continues to increase when \pw
decreases, though slightly. The expansion is indeed low but the
\watiw, which appears as a more dominant factor, decreases as the
number of threads inside the retry loop increases.

When looking into the differences between the constructive and the \avba
approach: the \avba approach estimations come out
to be less accurate for mid-contention cases as it only differentiates
between contended and non-contended modes. In addition, it fails to capture
the failing retries when measured throughput starts to deviate
from the theoretical upper bound, as \pw gets lower. In contrast, the
constructive approach provides high accuracy in all metrics for almost
every case.

We have also run the same synthetic tests with a \pww that follows a
Poisson distribution (Figure~\ref{fig:synt-poisson}) or is constant
(Figure~\ref{fig:synt-const}), in order to observe the impact of the
distribution nature of the \pww. Compared to the exponential
distribution, a better throughput is achieved with a Poisson
distribution on the \pww. The throughput becomes even better with a
constant \pww, since the slack time is minimized due to the
synchronization between the threads, as explained
in~\cite{our-disc15}.
}

\newcommand\synthtreib{
Here, we consider lock-free operations that can be completed
with a single successful \cas.
and provide predictions using both the \avba
and the constructive approach together with the theoretical upper bound.}


\newcommand\finemm{%
One quantum of the collection
phase is the collection of the list of one thread, while three nodes
are reclaimed during one quantum of the reclamation phase. The
traditional MM scheme was parameterized by a threshold based on the number
of the removed nodes; the fine-grain MM scheme
is parameterized by the number of quanta that are executed at each
call.

We apply different MM schemes on the \deqop operation of the
Michael-Scott queue, and plot the results in Figure~\ref{fig:mm_perf}.
We initialize the queue with enough elements. Threads execute
\deqop, which returns an element, then call the MM scheme.
On the left side, we compare a pure queue (without MM), a queue with
the traditional MM (complete reclamation once in a while) and a queue with
fine-grain MM (according to the numbers of quanta that are executed
at each call). Note that the
performance of the traditional MM is also subject to the tuning of the
threshold parameter. We have tested and kept only the best parameter
on the studied domain.
First, unsurprisingly, we can observe that the pure queue
outperforms the others as its \cw is lower (no need to maintain the
list of nodes that a thread is accessing).
Second, as the fine-grain MM is called after each completed \deqop,
adding a constant work, the MM can be seen as a part of the \pww. We
highlight this idea on the second experiment (on the right side). We first
measure the work done in a quantum. It follows that, for each value
of the granularity parameter, we are able to estimate the effective
\pww as the sum of the initial \pw and the work added by the
fine-grain MM. Finally, we run the queue with the fine-grain MM, and
plot the measured throughput, according to the effective parallel
work, together with our two approaches instantiated with the effective
\pw. The graph shows the validity of the model estimations for all
values of the granularity parameter.
}


\newcommand\adaptsine{
Numerous scientific applications are built upon a pattern of
alternating phases, that are communication- or
computation-intensive. If the application involves \dss, it is
expected that the rate of the modifications to the \dss is high
in the data-oriented phases, and conversely.
These phases could be clearly separated, but the application can also
move gradually between phases. The rate of modification to a \ds
will anyway oscillate periodically between two extreme values.
We place ourselves in this context, and evaluate the two MMs
accordingly. The \pww still follows an exponential distribution of
mean \pw, but \pw varies in a sinusoidal manner with time, in order to
emulate the numerical phases. More precisely, \pw is a step
approximation of a sine function. Thus, two additional parameters
rule the experiment: the period of the oscillating function represents
the length of the phases, and the number of steps within a period
depicts how continuous are the phase changes.
}


\newcommand\fulldeq{
We consider the \deq designed in~\cite{deq}. \pushl and
\pushr (resp. \popl and \popr) operations are exactly the same, except
that they operate on the different ends of the \deq.
The status flags, which depict the state of the \deq, and the
pointers to the leftmost element and the rightmost element are
together kept in a single double-word variable, so-called
\anch, which could be modified by a double-word \cas
atomically.

A \popl operation linearizes and even completes in one \staw that ends
with a double-word \cas that just sets the left pointer of the anchor
to the second element from left.

A \pushl operation takes three \staws to complete. In the first \staw,
the operation is linearized by setting the left pointer of the \anch
to the new element and at the same time changing the status flags to
``left unstable''\pr{.}{, to indicate the status of the incomplete but
linearized \pushl operation.} In the second \staw, the left pointer of
the leftmost element is redirected to the recently pushed element.
In the third \staw, a \cas is executed on \anch to bring the \deq
status flags into ``stable state''. Every operation can help an incomplete
\pushl or \pushr until the \deq comes into the stable state; in this
state, the other operations can attempt to linearize anew.

As noticed, the first and the third \staw execute a \cas on the same
variable (\anch) so it is possible to delay the third \staw of the
\supw by executing a \cas in the first \staw. This implies
that the expansion in \staw one should also be considered when the
delay in the third \staw is considered, and the other way around. This
can be done by summing expansion estimates of the \staws that run the
\cas on the same variable and using this expansion value in all these
\staws. Again, it just requires simple modifications in the expansion
formula by keeping assumptions unchanged.

We first run pop-only and push-only experiments where dedicated
threads operate on both ends of the \deq, in a half-half
manner. We provide predictions by plugging the slightly modified
expansion estimate, as explained above, into the \avba approach. Then,
we take one step further and mix the operations, assigning the threads
inequally among push and pop operations.
And, we obtain estimates for them by simply taking the weighted
average (depending on the number of threads running each operation) of
the \supw of pop-only and push-only experiments, with
the corresponding \pw value.

\begin{figure}[h!]
\begin{center}
\includegraphics[width=.8\textwidth]{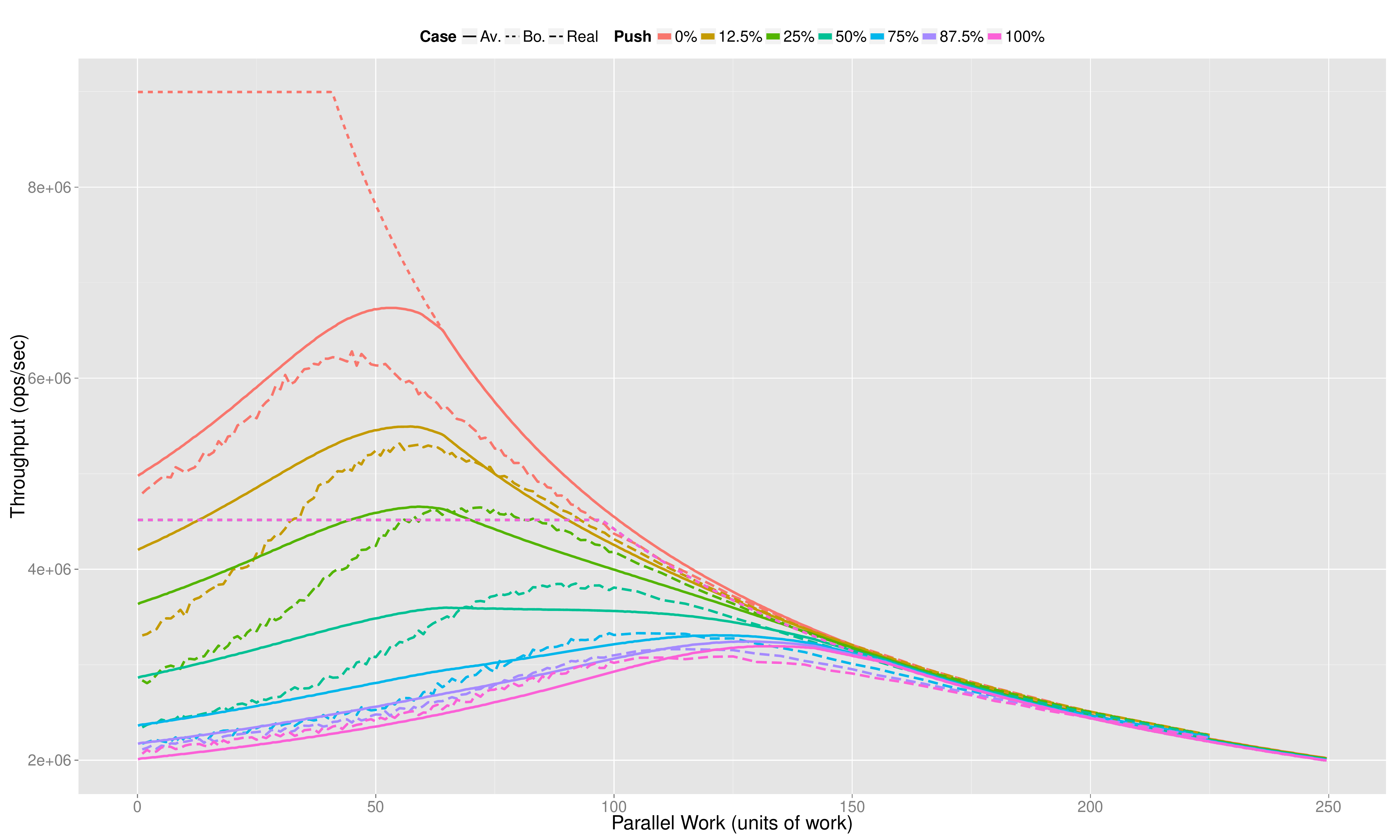}
\end{center}
\caption{Operations on \deq\label{fig:deq}}
\end{figure}

In Figure~\ref{fig:deq}, results are illustrated; they are
satisfactory for the push-only and pop-only cases.
For the mixed-case experiments, the results are mixed: our analysis
follows the trend and becomes less accurate
when the \pw gets lower, as experimental
curves tend toward push-only \supw. This, presumably, happens because
the first \staw of a \pushl (or \pushr) operation is shorter than the
first \staw of a \popl (or \popr) operation. This brings indeed an
advantage to push operations, under contention: they have higher chances
to linearize before pop operations after the \ds comes into the stable
state. It\rr{ also} provides an interesting observation which highlights
the lock-free nature of operations: it is improbable to complete a pop
operation if numerous threads try to push, due to the
difference of work inside the first \staw of their \rl.
}


\newcommand{\fullenq}{
As a first step, we consider the \enqop operation of the MS queue to
validate our approach.  This operation requires two
pointer updates leading to two \staws, each ending with a \cas. The
first \staw, that linearizes the operation, updates the next pointer
of the last element to the newly enqueued element. In the next and
last \staw, the queue's head pointer is updated to point to the recently
enqueued element, which could be done by a helping thread, that brings
the data structure into a stable state. Here, we determine the \cw by 
subtracting the \rc and \cc from the non-contended cost of \enqop operation.

\rr{
\begin{figure}[h!]
\begin{center}
\includegraphics[width=.8\textwidth]{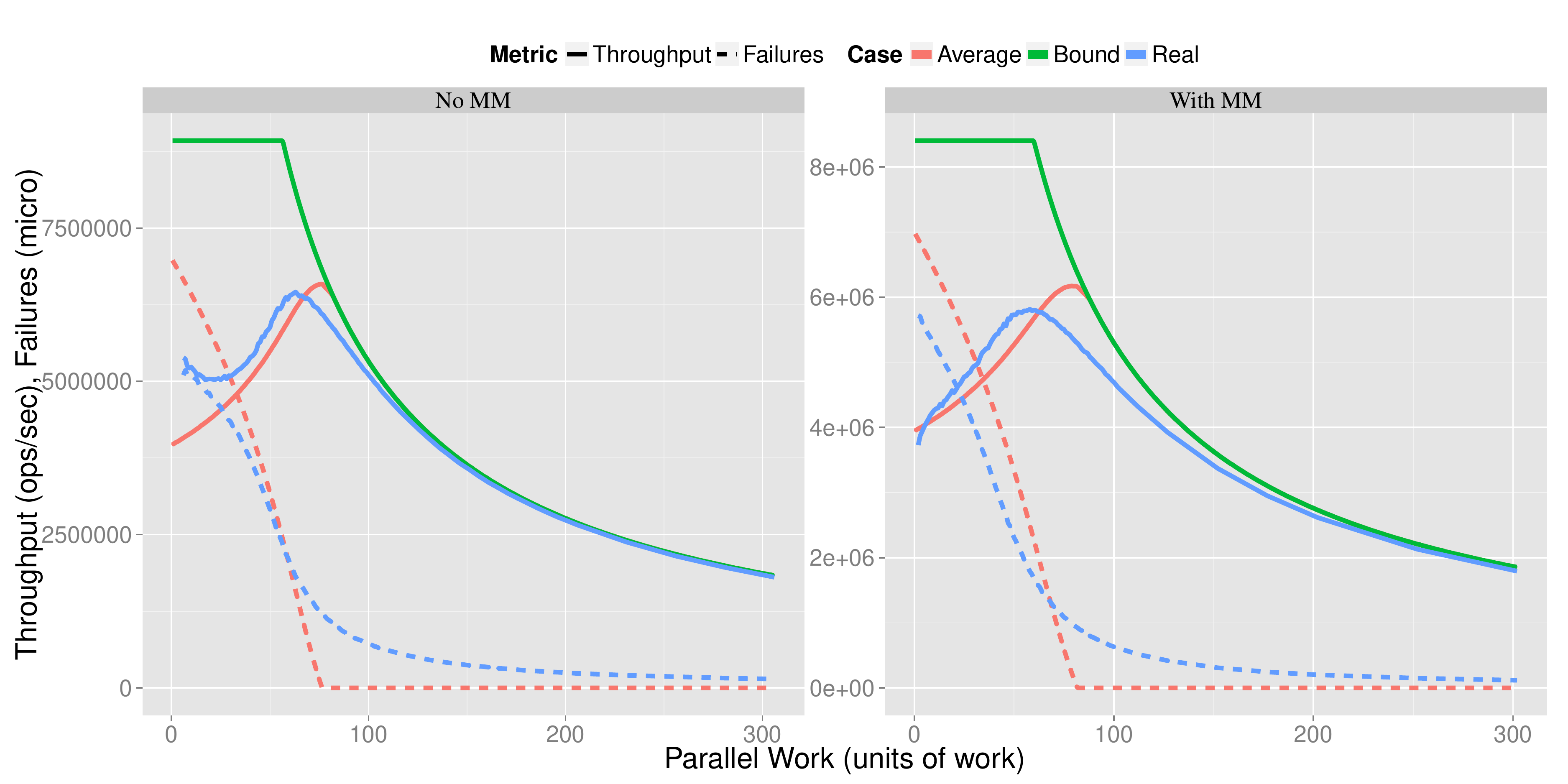}
\end{center}
\caption{Enqueue on MS Queue \label{fig:enqueue}}
\end{figure}

We estimate the expansion in the \supw as described above and
throughput as explained in Section~\ref{sec:avba}. The results for the
\enqop experiments where all threads execute \enqop are presented in
Figure~\ref{fig:enqueue}.
}
}


\newcommand{\fulladsexp}{
Consider an operation such that, the \supw (ignoring the slack time) is composed
of $S$ \staws (denoted by $\stag{1}, \dots, \stag{S}$) where each stage represents a step
towards the completion of the operation.
Let \casn{i} denote the \cas operation at the end of the \stag{i}.
From a system-wide perspective, $\{ \casn{1}, \dots , \casn{S}\}$ is the set of \cass that
have to be successfully and consecutively executed to complete an operation, assuming
all threads are executing the same operation.
This design enforces that \casn{i} can be successful only if the last successful \cas
is a \casn{i-1}. And, \casn{1} can be successful only if the last successful \cas
is a \casn{S}. In other words, another operation can not linearize before the
completion of the linearized but incomplete operation.

Now, let $e_i$ denote the expected expansion of \casn{i}.  If the \ds
is in the stable state (\ie is in \stag{1}, where a new operation can be
linearized), then we have to consider the probability, for all threads
except one, to expand the successful \casn{1} which linearizes the
operation. After the linearization, this operation will be completed
in the remaining stages where again the successful \cass at the end
of the stages are subject to the same expansion possibility by the
threads in the \rl, as they might be still trying to help for the
completion of the previously completed operation.

Similar to the~\cite{our-disc15}, our assumption here is that any
thread that is in the \rl, can launch \casn{i}, with probability $h$,
that might expand the successful \casn{i}. We consider, the starting
point of a failing \casn{i} is a random variable which is distributed
uniformly within the \rl, which is composed of expanded \staws of the
operation. This is because an obsolete thread can launch a \casn{i},
regardless of the \staw in which the \ds is in (equally, regardless of the
last successful \cas). Due to the uniformity assumption, the expansion
for the successful \cass in all stages, would be equal.  Similar to
the~\cite{our-disc15}, we estimate the expansion $e_i$ by considering the
impact of a thread that is added to the \rl. Let the cost function
$\mathit{delay}_i$ provide the amount of delay that the additional thread
introduces, depending on the point where the starting point of its
\casn{i} hits. By using these cost functions, we can formulate the
total expansion increase that each new thread introduces and derive
the differential equation below to calculate the expected total
expansion in a \supw , where $\avexp{\atrl}=\sum^{S}_{i=1}
\aexpi{\atrl}$. Note that, we assume that the expansion starts as soon
as strictly more than 1 thread are in the retry loop, in
expectation.

\begin{lemma}
\label{lem.1}
The expansion of a \cas operation is the solution of the following
system of equations, where $\rlw = \sum^{S}_{i=1} \rlw_i =
\sum^{S}_{i=1}(\rc_i + \cw_i + \cc_i)$:
\[ \left\{
\begin{array}{lcl}
\expansionp{\atrl} &=& \fcas \times \dfrac{S \times \frac{\fcas}{2} + \expansion{\atrl}}{ \rlw + \expansion{\atrl}}\\
\expansion{\trlo} &=& 0
\end{array} \right., \text{ where \trlo is the point where expansion begins.}
\]
\end{lemma}

\begin{proof}

We compute $\expansion{\atrl + h}$, where $h\leq1$, by assuming that
there are already \atrl threads in the \rl, and that a new thread
attempts to \cas during the \re, within a probability $h$. For
simplicity, we denote $a^i_j = (\sum_{j=1}^{i-1} \rlw_j + e_j(\atrl))
+ \rc_i + \cw_i$.

\begin{align*}
\expansion{\atrl + h}
&= \expansion{\atrl} + h\times
\sum^{S}_{i=1} \rint{0}{\rlsiz}{\frac {\shifti{i}{t_i}}{\rlsiz}}{t_i} \\
&= \expansion{\atrl}
     + h \times \sum^{S}_{i=1} \Big( \rint{0}{a^i_j - \cc}{\frac{\shifti{i}{t_i}}{\rlsiz}}{t_i}  +
      \rint{a^i_j - \fcas}{a^i_j}{\frac{\shifti{i}{t_i}}{\rlsiz}}{t_i}\\
& \quad\quad\quad\quad\quad\quad\quad\quad\quad\quad  + \rint{a^i_j}{a^i_j + \aexpi{\atrl}}{\frac{\shifti{i}{t_i}}{\rlsiz}}{t_i}
  + \rint{a^i_j + \aexpi{\atrl}}{\rlsiz}{\frac{\shifti{i}{t_i}}{\rlsiz}}{t_i}\Big)\\
  &= \expansion{\atrl} + h \times \sum^{S}_{i=1} \Big(
    \rint{a^i_j-\fcas}{a^i_j}{\frac{t_i}{\rlsiz}}{t_i}
      + \rint{a^i_j}{a^i_j + \aexpi{\atrl}}{\frac{\fcas}{\rlsiz}}{t_i} \Big)\\
\expansion{\atrl + h} &= \expansion{\atrl} + h \times \frac{ (\sum^{S}_{i=1} \frac{\fcas^2}{2}) + \expansion{\atrl}\times\fcas}{\rlsiz}
\end{align*}

This leads to
\[ \quad\frac{\expansion{\trl + h}- \expansion{\atrl}}{ h} = \frac{ S \times \frac{\fcas^2}{2} + \expansion{\atrl}\times\fcas}{\rlsiz}.\]
When making $h$ tend to $0$, we finally obtain
\[ \expansionp{\atrl} = \fcas \times \frac{S \times \frac{\fcas}{2} + \expansion{\atrl}}{ \rlw + \expansion{\atrl}}. \qedhere\]
\end{proof}

In addition, if a set $S_k$ of \cass are operating on the same
variable $var_k$, then $\casn{i} \in S_k$ can be expanded by the
$\casn{j} \in S_k$. In this case, we can obtain $\aexp{k}{\atrl}$ by
using the reasoning above. The calculation simply ends up as follows:
Consider the problem as if no \cas shares a variable and denote
expansion in \stag{i} with $\aexpi{\atrl}^{(\mathit{old})}$. Then, $\aexp{k}{\atrl}
= \sum_{\cas_i \in S_k} \aexpi{\atrl}^{(\mathit{old})}$.
}


\newcommand{\fulladswati}{
We assume here the slack time can only occur after the completion of
an operation (\ie before stage 1), as the other stages are expected to
start immediately due to the thread that completes the previous
stage. Similar to Section~\ref{sec:litt-slack}, we consider that, at
any time, the threads that are running the retry loop have the same
probability to be anywhere in their current retry.  Thus, a thread can
be in any stage just after the successful CAS that completes the
operation. So, we need to consider the thread which is closest to the
end of its current stage when the operation is completed. We denote
the execution time of the expanded retry loop with \rlsiz and the
number of \staws with $S$. For a thread executing \stag{i} when the
operation completes, the time before accessing the \ds is then
uniformly distributed between 0 and $\rlsiz_i$.

Here, we take another assumption and consider all stages can be
completed in the same amount of time (\ie for all (i, j) in $\{1,
\dots ,S\}^2$, $\rlsiz_i = \rlsiz_j = \rlsiz/S$). This assumption
does not diverge much from the reality and provides a reasonable
approximation. With these assumption and using
Lemma~\ref{lem:unif-min}, we conclude that:

\begin{equation}
\label{eq:slack-multiple}
\avwati{\atrl} = \frac{\rlsiz}{S \times (\atrl +1)}.
\end{equation}

}


\newcommand{\sumads}{
Here, we consider \dss that apply immediate helping, where threads
help for the completion of a recently linearized operation until the
\ds comes into a stable state in which a new operation can be
linearized. The crucial observation is that the \ds goes through
multiple \staws in a round robin fashion.  The first \staw is the one
where the operation is linearized. The remaining ones are the \staws
in which other threads, that execute another operation, might help for
the completion of the linearized operation, before attempting to
linearize their own operations.  Thus, the \supw (ignoring the \watiw)
can be seen as the sum of the execution time of these \staws, each
ending with a \cas that updates a pointer. The \cas in the first stage
might be expanded by the threads that are competing for the
linearization of their operation, and consequent \cass might be
expanded by the helper threads, which are still trying to help an
already completed operation.  Also, there might be slack time before
the start of the first \staw as the other \staws will start
immediately due to the thread that has completed the previous \staw.

Although it is hard to stochastically reconstruct the executions
with Markov chains, our \avba approach
provides the flexibility required to estimate the performance by
plugging the expected \supw, given the number of threads inside the
\rl, into the Little's Law.  As the impacting factors are similar, we
estimate the \supw in the same vein as in Section~\ref{sec:avba}; with
a minor adaptation of the expansion formula
and by slightly adapting the slack time estimation based on the
same arguments.
}


\newcommand\bothmm{
\falseparagraph{Fine-grain Memory Management Scheme}
We divide the routine (and further the phases) of the traditional MM
mechanism into quanta (equally-sized chunks).%
\finemm

\falseparagraph{Adaptive Memory Management Scheme}
We build the adaptive MM scheme on top of the fine-grain MM mechanism by
adding a monitoring routine that tracks the number of failed \rls,
employing a sliding windows. Given a granularity parameter and a
number of failed \rls, we are able to estimate the parallel work and
the throughput, hence we can decide a change in the granularity
parameter to reach the peak performance. Note that one can avoid
memory explosion by specifying a threshold like the traditional
implementation in case the application provides a durable low
contention; in the worst case, it performs like the traditional MM.
}


\newcommand\fullbosy{
In Figure~\ref{fig:bo-synt}, we compare, on a synthetic workload, this
constant back-off strategy against widely known strategies, namely
exponential and linear, where the back-off amount increases
exponentially or linearly after each failing retry loop starting from
a \cycles{115} step size.
}


\newcommand\figsynthcst{
\begin{figure}[b!]
\begin{center}
\includegraphics[width=.85\textwidth]{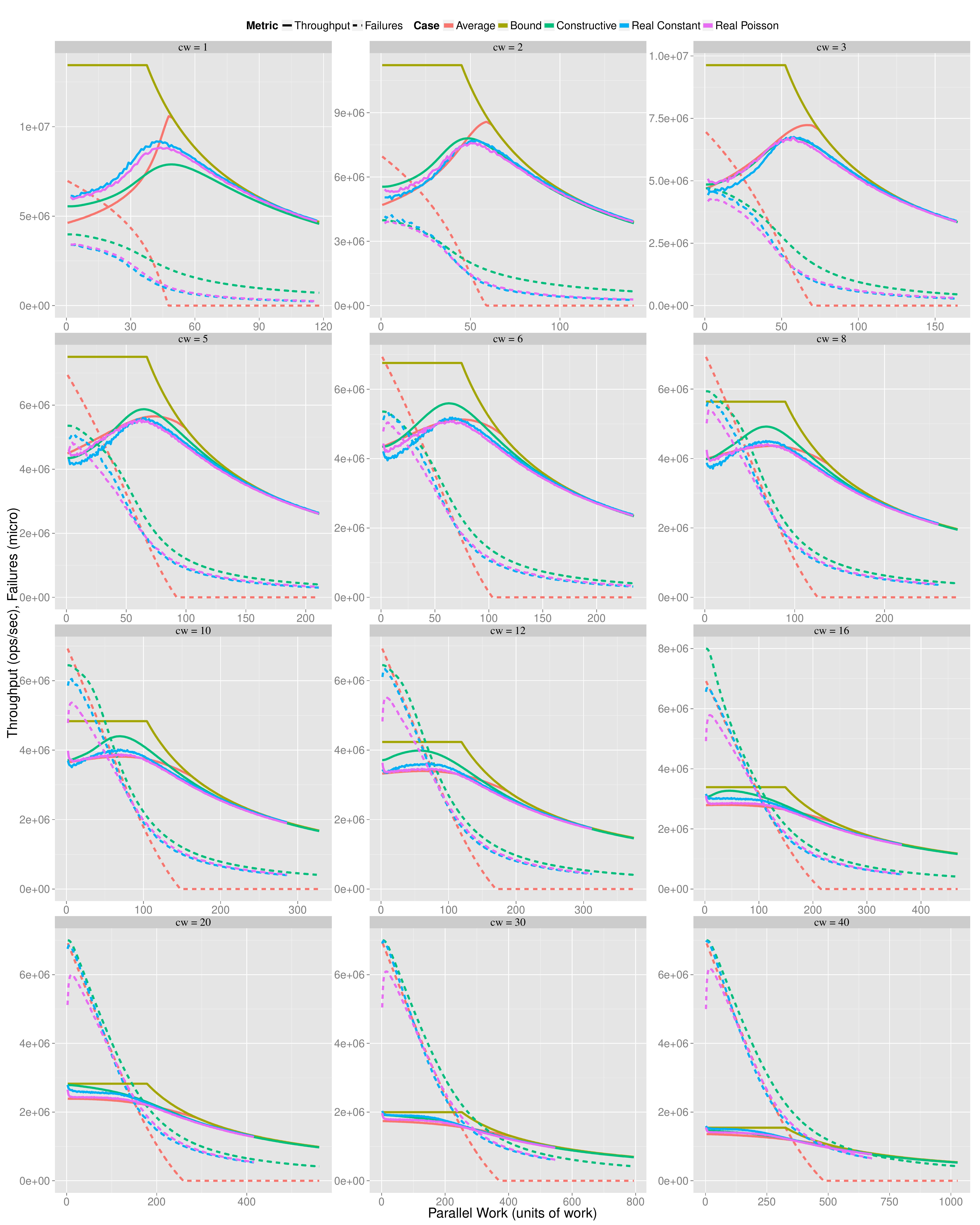}
\end{center}
\caption{Synthetic program with exponentially distributed \pww\label{fig:synt}}
\end{figure}}

\newcommand\figsynthpoi{
\begin{figure}[b!]
\begin{center}
\includegraphics[width=.85\textwidth]{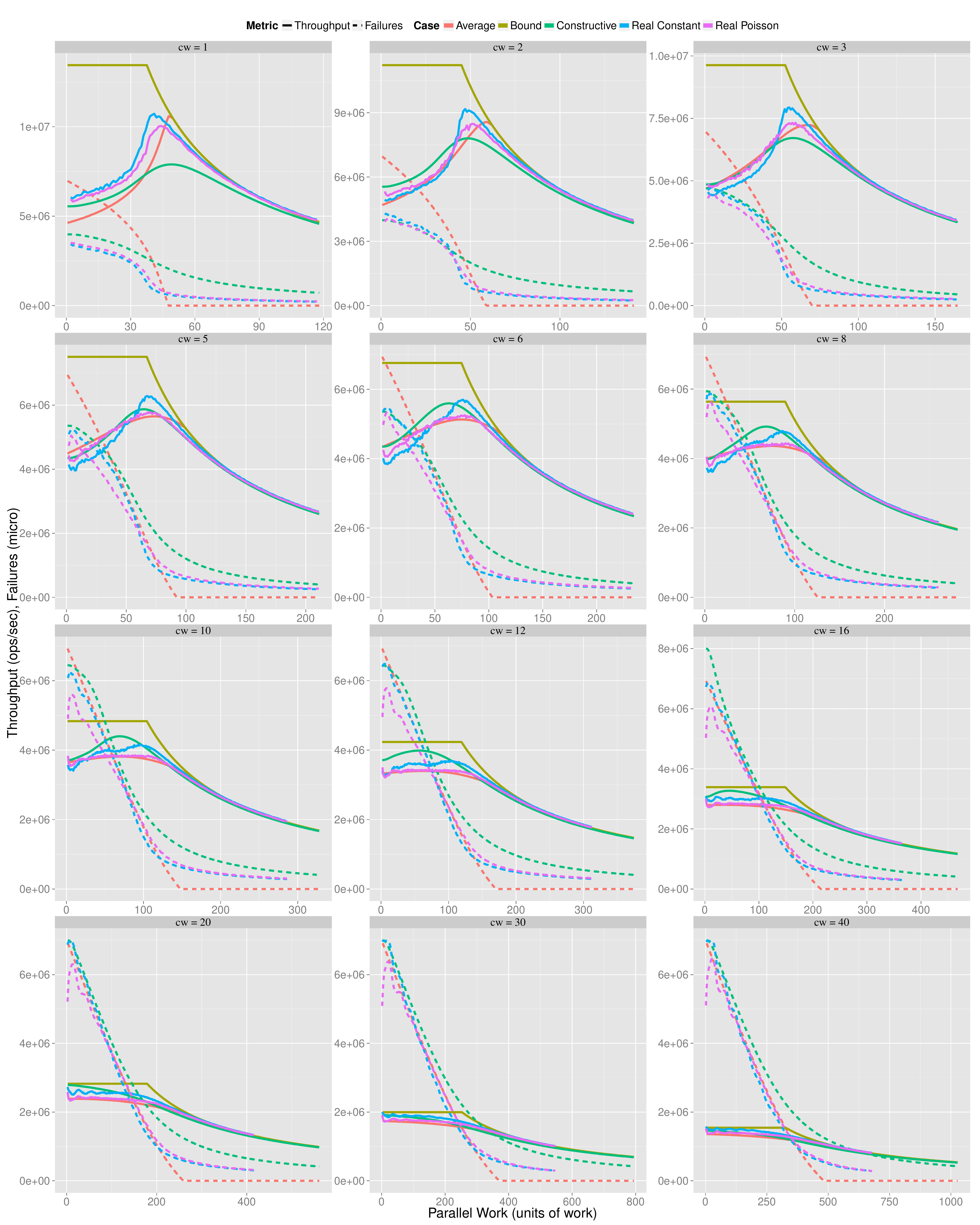}
\end{center}
\caption{Synthetic program with \pww following Poisson\label{fig:synt-poisson}}
\end{figure}}

\newcommand\figsynthconst{
\begin{figure}[b!]
\begin{center}
\includegraphics[width=.85\textwidth]{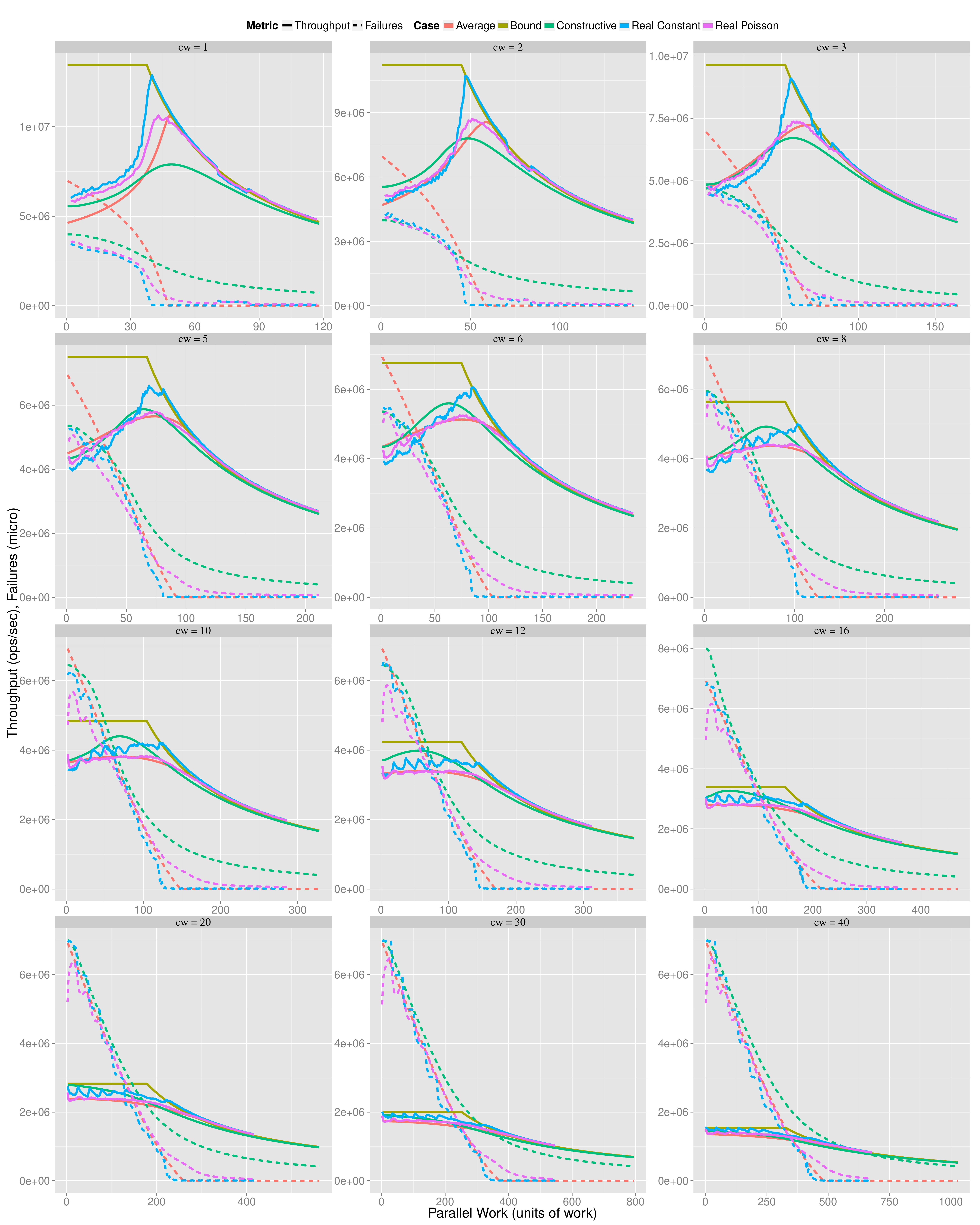}
\end{center}
\caption{Synthetic program with Constant \pww \label{fig:synt-const}}
\end{figure}}

\newcommand\figtreib{
\begin{figure}[h!]
\begin{center}
\includegraphics[width=\textwidth]{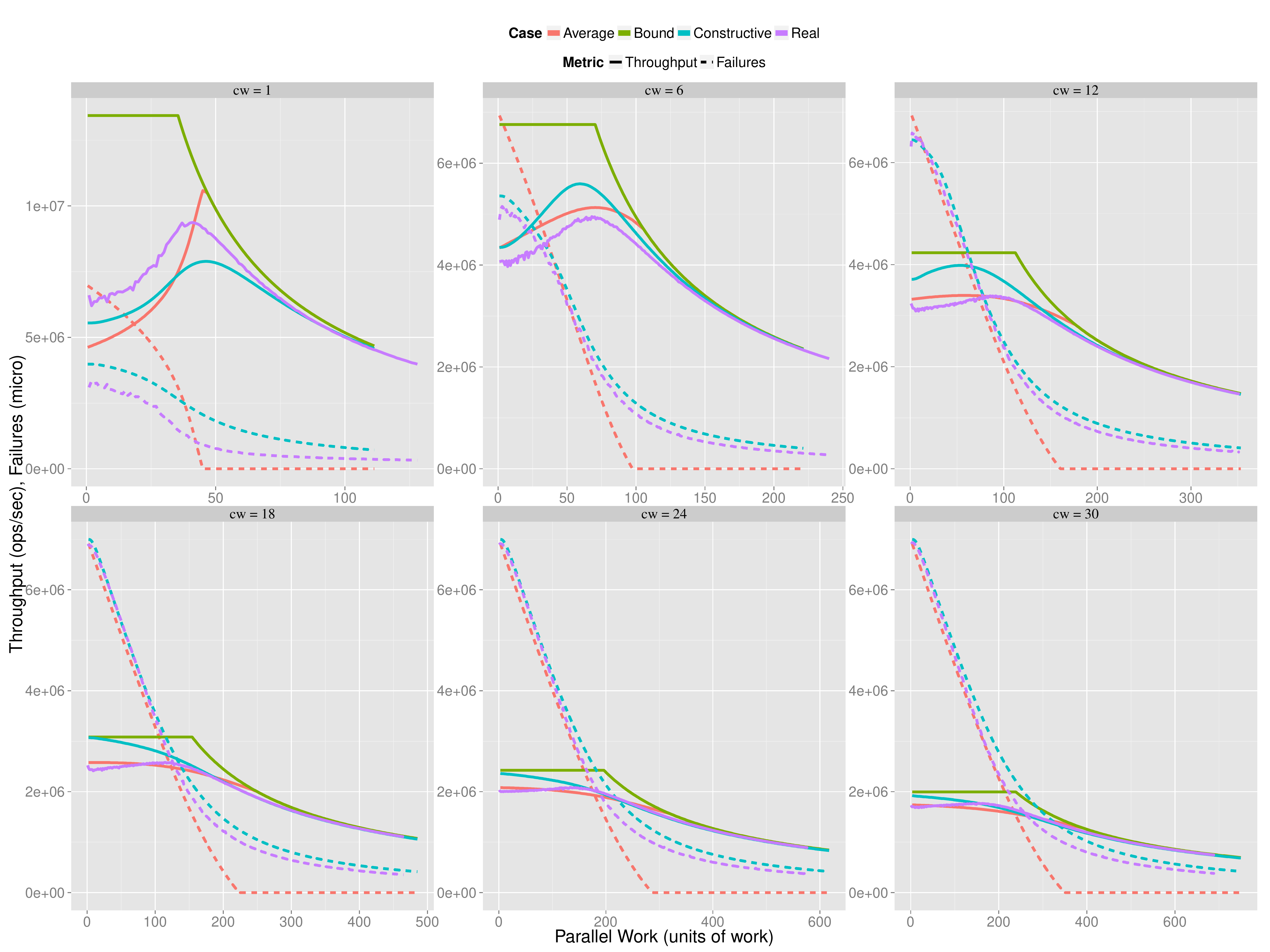}
\end{center}
\caption{Treiber's Stack\label{fig:stack}}
\end{figure}}

\newcommand\figcompmm{
\begin{figure}[b!]
\begin{center}
\pr{\includegraphics[width=.8\textwidth]{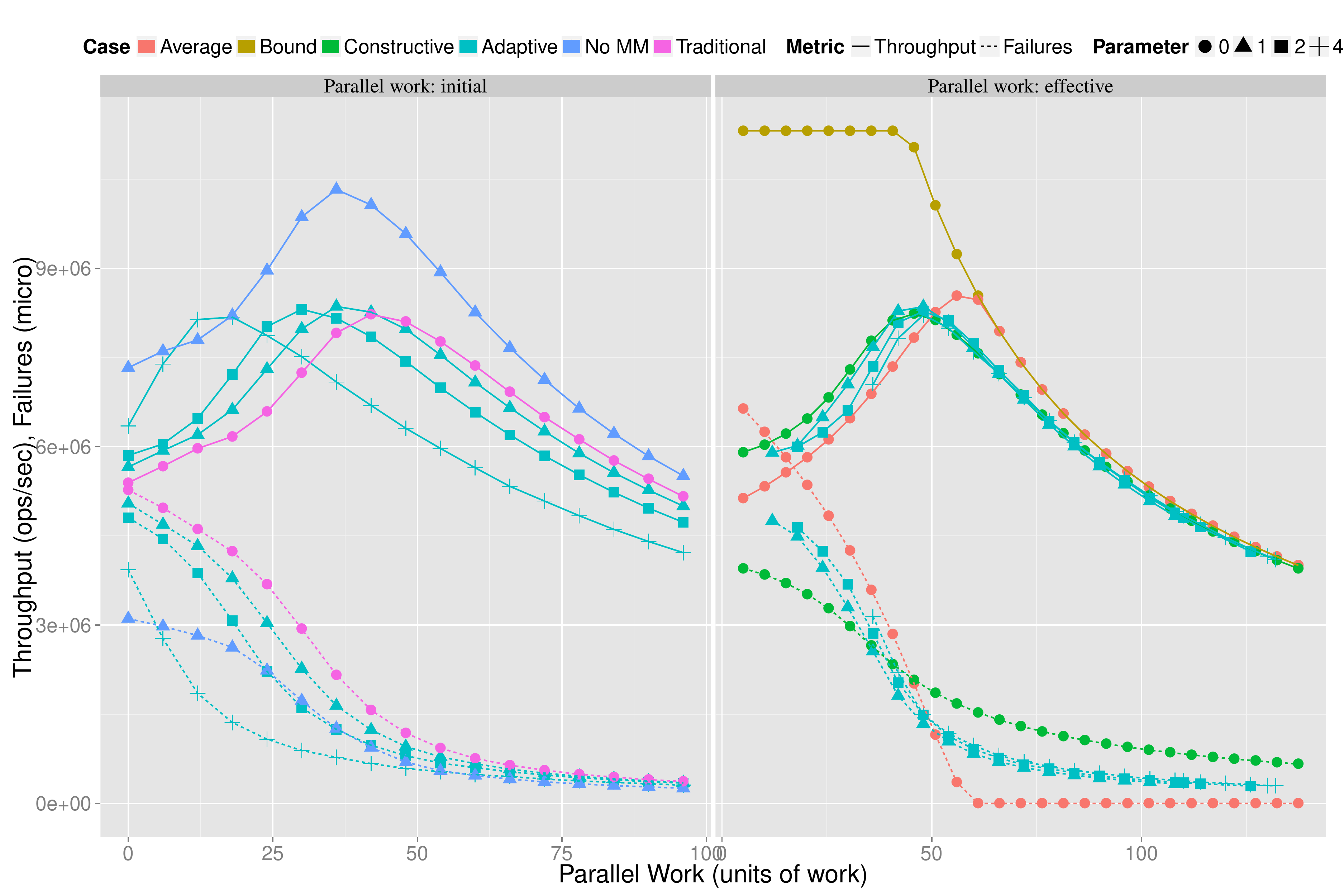}}{\includegraphics[width=.9\textwidth]{dequeue_xp_rr_disc}}
\end{center}
\caption{Performance of memory management mechanisms\label{fig:mm_perf}}
\end{figure}}


\section{Introduction}

During the last two decades, lock-free \dss have received a lot of
attention in the literature, and have been accepted in industrial
applications, \eg in the Intel's Threading Building Blocks
Framework~\cite{itbbf}, the Java concurrency package~\cite{jav-conc}
and the Microsoft .NET Framework~\cite{mic-net-f}.
\rr{Lock-free implementations provide indeed a way out of several
limitations of their lock-based counterparts, in robustness,
availability and programming flexibility. Last but not least, the
advent of multi-core processors has pushed lock-freedom on top of the
toolbox for achieving scalable synchronization.}

Naturally, the development of lock-free \dss was accompanied by
studies on the performance of such \dss, in order to characterize their
scalability.
Having no guarantee on the execution time of an individual operation,
the time complexity analyses of lock-free algorithms have turned
towards amortized analyses.
The so-called amortized analyses are thus interested in the \wc
behavior over a sequence of operations, which can be seen as a \wc
bound on the average time per operation.
In order to cover various contention environments, the time complexity
of the algorithms is often parametrized by different contention
measures, such as point~\cite{point-contention}, interval~\cite{interval-contention}
or step~\cite{step-contention} contention.
Nonetheless these investigations are targeting worst-case asymptotic
behaviors.  There is a lack of analytical results in the literature
capable of describing the execution of lock-free algorithms on top of a
hardware platform, and providing predictions that are close to what is
observed in practice.
Asymptotic bounds are particularly useful to rank different
algorithms, since they rely on a strong theoretical background, but
the presence of potentially high constants might produce misleading
results. Yet, an absolute prediction of the performance can be of
great importance by constituting the first step for further
optimizations. 

The common measure of performance for \dss is throughput, defined
as the number of operations on the \ds per unit
of time.
To this end, this performance measure is usually obtained by
considering an algorithm that strings together a pure sequence of
calls to an operation on the \ds. However, when used in a more
realistic context, the calls to the operations are mixed with
application-specific code (that we call here \pww). For instance, in a
work-stealing environment designed with \deqs, a thread basically runs
one of the following actions: pushing a new-generated task in its
\deq, popping a task from a \deq or executing a task. The
modifications on the \deqs are thus interleaved with \deq-independent
work. There exist some papers that consider in their experiments local
computations between calls to operations during their respective
evaluations, but the amount of local computations follows a given
distribution varying from paper to paper, \eg constant~\cite{lf-queue-michael}, uniform \cite{scalable-stack-uniform},
exponential \cite{Val94}.

In this work, we derive a general approach for unknown distributions of 
the size of the application-specific code, as well as a tighter method 
when it follows an exponential distribution.

As for modeling the \ds itself, we use as a basis the universal construction
described by Herlihy in~\cite{herli-univ-const}, where it is shown
that any abstract data type can get such a lock-free implementation, which
relies on one \rl.
Moreover, we have particularly focused our experiments on \dss
that present a low level of disjoint-access
parallelism~\cite{disjoint-access} (stack, queue, shared counter,
\deq). Coming back to amortized analyses, the time complexity of an
operation is often expressed as a contention-free time complexity
added with a contention overhead. In this paper, we want to model and
analyze the impact of contention.
Loosely speaking, the \dss that exhibit low level of disjoint-access parallelism 
have lightweight operations (\ie low contention-free complexity) 
and they are prone to high contention overheads. In contrast, the \dss 
that present high level of disjoint-access parallelism, or 
that employ 
contention alleviation techniques, provide heavyweight operations 
(\ie high contention-free complexity) and behave differently, compared to the
previous ones, under contention.
Our analyses examine this trade-off and then facilitate conscious decisions in the 
\dss design and use.

We propose two different approaches that analyze the performance of
such \dss.
On the one hand, we derive an \avba approach invoking queuing theory,
which provides the throughput of a lock-free algorithm without any
knowledge about the distribution of the \pww. This approach
is flexible but allows only a coarse-grained analysis, and hence a
partial knowledge of the contention that stresses the \ds.
On the other hand, we exhibit a detailed picture of the execution of
the algorithm when the \pww is instantiated with an exponential
distribution, through a second complementary approach. We prove that
the multi-threaded execution follows a Markovian process and a Markov
chain analysis allows us to pursue and reconstruct the execution, and
to compute a more accurate throughput.

We finally show several ways to use our analyses and we evaluate the
validity of our ideas by experimental results. Those two analysis
approaches give a good understanding of the phenomena that drive the
performance of a lock-free \ds, at a high-level for the \avba
approach, and at a detailed level for the constructive
method. Moreover, our results provide a common framework to compare
different implementations of a \ds, in a fair manner. We also
emphasize that there exist several concrete paths to apply our
analyses.
To this end, based on the knowledge about the application at hand, we
implement two back-off strategies. We show the applicability of these strategies by tuning a Delaunay triangulation
application~\cite{caspar} and a streaming pipeline component which is
fed with trade exchange workloads~\cite{taq-se}.
These experiments reveal the validity of our analyses in the application 
domain, under non-synthetic workloads and diverse access patterns.
We
confirm the benefits of our theoretical results by designing
a new adaptive memory management mechanism for
lock-free data structures in dynamic environments which surpasses the
traditional scheme and which is such that the loss in performance,
when compared to a static data structure without memory management, is
largely leveraged. This memory management mechanism is based on the analyses presented in this paper. 

\rr{
The rest of the paper is organized as follows: we start by presenting
related work in Section~\ref{sec:related}, then we define the
algorithm and the platform that we consider, together with concepts
that are common to our both approaches in Section~\ref{sec:preli}. The
\avba approach is described in Section~\ref{sec:avba}, while the
constructive analysis is exposed in Section~\ref{sec:cons}, and both
methods are evaluated in the experiment part that is presented in
Section~\ref{sec:xp}.
}

\vspp{-.4}
\section{Related Work}
\label{sec:related}

\rr{Approaches that are based on Markov chains and queueing theory, are
commonly employed to analyze the performance of parallel programs
in concurrent environments.
Yu \etal~\cite{yu-markov} have provided an analytical model to
estimate the mean transaction completion time for the transactional
memory systems. They make use of a continuous-time Markov chain
queuing model to analyze the execution of transactions, in which
they formulate the state transition rates by considering the arrival
rate, the service time for the transactions together with other parameters
such as conflict rate that statistically quantifies the spatial
(intersecting data set) and temporal (overlapped time) aspects of
conflicts.
\rr{In~\cite{bahra}, Al{-}Bahra has mentioned
Little's Law as an appropriate tool to understand the effects of
contention on serialized resources for synchronization paradigms.}

Closer to our work, }Alistarh \etal~\cite{ali-same} have studied the
same class of lock-free \dss that we consider in this
paper.
They show initially that the lock-free algorithms
are statistically wait-free and going further they exhibit upper bounds
on the performance.
Their analysis is done in terms of scheduler steps, in a system where
only one thread can be scheduled (and can then run) at each step. If
compared with execution time, this is particularly appropriate to a
system
where the instructions of the threads cannot be done in parallel (\eg
multi-threaded program on a multi-core processor with only
writes on the same cache line of the shared memory). In our paper, the
execution is evaluated in terms of processor cycles, strongly related
to the execution time. In addition, the ``\pww'' and the ``\cww'' can
be done in parallel. Also,
\rr{they bound the asymptotic expected system
latency (with a big O, when the number of threads tends to
  infinity), while }%
in our paper we estimate the throughput (close to
the inverse of system latency) for any number of threads.

\textbf{\textit{Comparing to our previous work:}} In~\cite{our-disc15}, we illustrate the 
performance impacting 
factors and the model we use to cover a subset of lock-free structures 
that we consider in this paper. In the former paper, the analysis is built upon
properties that arise only when the sizes of the \cww and the \pww are
\textit{constant}. There, we show that the execution is not memoryless due to the
natural synchrony provided by the \rls; at the end of the line, we
prove that the execution is cyclic and use this property to bound the
rate of failed \res.

Here, we provide two new approaches which serve different purposes. In the first 
approach, we relax the assumptions regarding the \cww and \pww parameters, that we 
respectively use to model the \ds operations and the application specific code 
from which the \ds operations are called. The first approach relies on the expected values of 
the size of the \cww and the \pww. This allows us to cover, compared to our previous analysis, more
advanced 
lock-free \ds operations, see Section~\ref{sec:advanced-ds}. Also, we can analyse the \dss running on 
a larger variety of application specific environments, thanks to the relaxed 
assumption on the size of the \pww. 
The second approach provides a tight analysis when the \pww follows an exponential 
distribution. We can observe a significant decrease in the performance when 
the \pww is initiated with exponential distribution in comparison to the cases 
where the \pww is constant as in our previous work, see \pr{Appendix~\ref{app:xp-basic}}{Section~\ref{sec:synt_tests}}. The 
tight analyses, in our previous work and the second approach presented in this paper, reveal 
for the first time an analytical understanding of this phenomenon.

This paper is complementary to the previous work, not only because of the 
difference in the analysis tools, the extensive set of \dss and the application specific 
environments that it considers
but also because they together exhibit the impact of the size distributions 
of the \pww on the performance of lock-free \dss.

\vspp{-.15}
\section{Preliminaries}
\label{sec:preli}

We describe in this section the structure of the algorithm that is covered by
our model. We explain how to analyze the execution of an instance of
such an algorithm when executed by several threads, by slicing this
execution into a sequence of adjacent \supws, where a \supw is an interval of
time during which exactly one operation returns. Each of the \supws
is further split into two by the first access to the \ds in the
considered \rl.
This execution pattern reflects fundamental phases of both analyses,
whose first steps and general direction are outlined at the end of the
section.

\vspp{-.2}
\subsection{System Settings}

All threads call Procedure~\ref{alg:gen-name} (see
Figure~\ref{alg:gen-nb}) when they are spawned. So each thread follows
a simple though expressive pattern: a sequence of calls to an
operation on the \ds, interleaved with some \pww during which the
thread does not try to modify the \ds.
For instance, it can represent a work-stealing algorithm, as
described in the introduction.

The algorithm is decomposed in two main sections: the {\it \ps},
represented on line~\ref{alg:li-ps}, and the {\it \rl} (which
represents one operation on the shared \ds) from
line~\ref{alg:li-bcs} to line~\ref{alg:li-ecs}. A {\it \re} starts at
line~\ref{alg:li-bbcs} and ends at line~\ref{alg:li-ecs}. The outer
loop that goes from line~\ref{alg:li-bwl} to line~\ref{alg:li-ecs} is
designated as the {\it \wl}\rr{.

}\pp{. }
In each \re, a thread tries to modify the \ds and does not exit the
\rl until it has successfully modified the \ds.
\rr{It firstly reads the
access point \DataSty{AP} of the \ds, then, according to the value
that has been read, and possibly to other previous computations that
occurred in the past, the thread prepares, during the \cww,
the new desired value as an access point of
the \ds. Finally, it atomically tries to perform the change through a
call to the \cas primitive. If it succeeds, \ie if the access point
has not been changed by another thread between the first \rf and the
\cas, then it goes to the next \ps, otherwise it repeats the
process. }%
The \rl is composed of at least one \re (and the first
iteration of the \rl is strictly speaking not a \re, but a try).

We denote by \cc the execution time of a \cas{} when the executing
thread does not own the cache line in exclusive mode, in a setting
where all threads share a last level cache. Typically, there
exists a thread that touches the data between two requests of the same
thread, therefore this cost is paid at every occurrence of a \cas.
As for the \rf{}s, \rc holds for the execution time of a cache
miss. When a thread executes a failed \cas, it immediately reads the
same cache line (at the beginning of the next \re), so the cache line
is not missing, and the execution time of the \rf is considered as
null. However, when the thread comes back from the \ps, a cache miss
is paid.
To conclude with the parameters related to the platform, we dispose of
\ct cores, where the \cas (resp. the \rf) latency is identical for all
cores, \ie \cc (resp. \rc) is constant.

The algorithm is parametrized by two execution times. In the general
case, the execution time of an occurrence of the \ps
(application-specific section) is a random variable that follows an
unknown probability distribution. In the same way, the execution time
of the \cww (specific to a \ds) can vary while following an unknown
probability distribution. The only provided information is the mean
value of those two execution times: \cw for the \cww, and \pw for the
\pww. These values will be given in units of work, where $\uow{1} =
\cycles{50}$.\vspp{-.3}

\subsection{Execution Description}
\label{sec:fra-exe}

It has been underlined in~\cite{our-disc15} that there are two
main conflicts that degrade the performance of the \dss which do not
offer a great degree of disjoint-access parallelism: logical and
hardware conflicts.

{\it Logical conflicts} occur when there are more than one thread
in the retry loop at a given time (happens typically when the
number of threads is high or when the parallel section is small).
At any time, considering only the threads
that are in the \rl, there is indeed at most one thread whose \re will
be successful (\ie whose ending \cas will succeed), which implies the
execution of more \res for the failing threads.
In addition, after a thread executes successfully its final
\cas, the other threads of the \rl have first to finish their current
\re before starting a potentially successful \re, since they are not
informed yet that their current \re is doomed to failure. This creates
some ``holes'' in the execution where all threads are executing useless
work.

The threads will also experience {\it hardware conflicts}: if several
threads are requesting for the same data, so that they can operate a
\cas on it, a single thread will be satisfied. All the other threads
will have to wait until the current \cas is finished, and give a new
try when this \cas is done. While waiting for the ownership of the
cache line, the requesting threads cannot perform any useful
work. This waiting time is referred to as {\it expansion}.

\def\herec{.5}
\def\wirec{2}
\def\grey{black!20}
\def\greywh{black!4}

\def\maroon{blue!20!black!40!red!}
\def\green{black!20!green}

\newcommand{\pha}[5]{%
\node[it,text width=#4em,right= 0 of #2,fill=#5] (#1) {#3};
}

\newcommand{\supfig}{
\begin{center}
\begin{tikzpicture}[%
it/.style={%
    rectangle,
    text width=11em,
    text centered,
    minimum height=3em,
    draw=black!50,
    scale=.85
  }
]
\coordinate (O) at (0,0);
\pha{pcas}{O}{successful\\\cas}{\pr{4.5}{5}}{\green}
\pha{sla}{pcas}{useless\\work}{\pr{4}{8}}{\grey}
\pha{acc}{sla}{\acc}{\pr{3.5}{4}}{orange}
\pha{cri}{acc}{\cw}{\pr{2}{4}}{\maroon}
\pha{exp}{cri}{expansion}{\pr{4.5}{6}}{\grey}
\pha{fcas}{exp}{successful\\\cas}{\pr{4.5}{5}}{\green}
\draw [decorate,decoration={brace,mirror,amplitude=10pt}]
(sla.south west) -- (sla.south east)
node [black,midway,yshift=-20pt] {\watiw};
\draw [decorate,decoration={brace,mirror,amplitude=10pt}]
(acc.south west) -- (fcas.south east)
node [black,midway,yshift=-20pt] {\compw};
\draw [decorate,decoration={brace,amplitude=10pt}]
(sla.north west) -- (fcas.north east)
node [black,midway,yshift=20pt] (supw) {\supw};

\node (anot) at (.8,1) {{\scriptsize can be null}};

\draw[very thin] ($(anot.south east)!.8!(anot.east)$) -- ++(.4,0) -- ($(sla.north)+(0,-.1)$);
\draw[very thin] ($(anot.north east)!.8!(anot.east)$) -- ++(.4,0) -- ($(exp.north)+(0,-.1)$);

\end{tikzpicture}
\end{center}
}

\rr{\begin{figure}[t!]
\abstalgo
\caption{Thread procedure}\label{alg:gen-nb}
\end{figure}}

\rr{\begin{figure}[t!]
\supfig
\captionof{figure}{\Supw\label{fig:seq}}
\end{figure}}

\pp{\begin{figure}[t!]
\centering
\begin{minipage}{.4\textwidth}
\abstalgo
\captionof{figure}{Thread procedure\label{alg:gen-nb}}
\end{minipage}\hfill%
\begin{minipage}{.6\textwidth}
\supfig
\captionof{figure}{\Supw\label{fig:seq}}
\end{minipage}
\end{figure}}

We now refine the description of the execution of the algorithm. The
timeline is initially decomposed into a sequence of \supws that will
define the throughput. A \supw is an interval of time of the
execution that
%
(i) starts after a successful \cas,
(ii) contains a single successful \cas,
(iii) finishes after this successful \cas.
\pr{To}{As explained in the previous subsection, to} be successful in its \re,
a thread has first to access the \ds, then modify it locally, and
finally execute a \cas, while no other thread performs changes on the
\ds. That is why each \supw is further cut into two main phases (see
Figure~\ref{fig:seq}). During the first phase, whose duration is
called the {\it \watiw}, no thread is accessing the \ds. The second
phase, characterized by the {\it \compw}, starts with the first access
to the \ds (by any thread). Note that this \acc could be either a \rf
(if the concerned thread just exited the \ps) or a failed \cas (if the thread
was already in the \rl).
The next successful \cas will come at least
after \cw (one thread has to traverse the \cww anyway), that is why we
split the latter phase into: \cw, then expansion, and finally a
successful \cas.

\subsection{Our Approaches}
\label{sec:fra-app}

In this work, we propose two different approaches to compute the
throughput of a lock-free algorithm, which we name as \avba and
constructive. The \avba approach relies on queuing theory and is
focused on the average behavior of the algorithm: the throughput is
obtained through the computation of the expectation of the \supw at a
random time.
As for the constructive approach, it describes precisely the instants
of accesses and modifications to the \ds in each \supw: in this way,
we are able to deconstruct and reconstruct the execution, according to
observed events. The constructive approach leads to a more accurate
prediction at the expense of requiring more information about the
algorithm: the distribution functions of the critical and \pwws have
indeed to be instantiated.

In both cases, we partition the domain space into different levels of
contention (or {\it modes}); these partitions are independent across
approaches, even if we expect similarities, but in each case, cover
the whole domain space (all values of \cww, \pww and number of
threads).
\medskip

\subsubsection{\Avba Analysis}
\label{sec:fra-app-asy}

We distinguish two main modes in which the algorithm can run:
contended and non-contended. In the non-contended mode, \ie when the
\pww is large or the number of threads is low,
concurrent operations are not likely to collide. So every \rl will
count a single \re, and atomic primitives will not delay each
other. In the contended mode, any operation is likely to experience
unsuccessful \res before succeeding (logical conflicts), and a \re
will last longer than in the non-contended mode because of the
collision of atomic primitives (hardware conflicts).

Once all the parameters are given, 
the analysis is centered around the calculation of a single variable
\atrl, which represents the expectation of the number of threads
inside the \rl at a random instant. Based on this variable, we are
able to express the expected expansion \avexp{\atrl} at a random
time. As a next step, we show how this expansion can be used to
estimate the expected \watiw \avwati{\atrl} and the expected \compw
\rwh{\atrl}, and at the end, the expected time of a \supw
\avsupe{\atrl}.

\subsubsection{Constructive Method}
\label{sec:fra-app-con}

The previous \avba reasoning is founded on expected values at a random time,
while in the constructive approach, we study each \supw individually,
based on the number of threads at the beginning of the considered
\supw. So we are able to exhibit more clearly the instants of
occurrences of the different accesses and modifications to the \ds,
and thus to predict the throughput more accurately.

We rely on the same set of values used in the \avba approach, but
these values are now associated with a given
\supw.
Thus the number of threads inside the \rl \trl, as well as the \watiw
and the \compw are evaluated at the beginning of each \supw.
We denote these times in the same way as in the first approach, but
remove the bar on top since these values are not expectations any
more.

The different contention modes do not characterize here the
steady-state of the \ds as in the previous approach but are
associated with the current \supw. Accordingly, the contention can
oscillate through different modes in the course of the execution.
First, a \supw is not
contended when $\trl=0$, \ie when there is no thread in the \rl after
a successful \cas.
In this case, the first thread that exits the \ps
will be successful, and the \acc of the sequence will be a \rf.
Second, the contention of a \supw is high when at any time during
the \supw, there exists a thread that is executing a \cas. In other
words, at the end of each \cas, there is at least one thread that is
waiting for the cache line to operate a \cas on it. This implies that
the first access of the \supw is a \cas and occurs immediately after
the preceding successful \cas: the \watiw is null.
Third, the mid-contention mode takes place when $\trl>0$, while at the
same time, there are not enough requesting threads to fill the whole
\supw with \cass (which implies a non-null \watiw). Since these
requesting threads have synchronized in the previous \supw, \cass do
not collide in the current \supw, and because of that, the expansion
is null.

\section{Average-based Approach}
\label{sec:avba}

We propose in this section our coarse-grained analysis to predict the
performance of lock-free \dss.
Our approach utilizes
fundamental queuing theory techniques, describing the average
behavior of the algorithm.
In turn, we need only a minimal knowledge about the algorithm: the mean
execution time values \cw and \pw.
As explained in Section~\ref{sec:fra-app-asy}, the system runs in one of
the two possible modes: either contended or uncontended.

\subsection{Contended System}

We first consider a system that is contended.
When the system is contended, we use Little's law to obtain, at a
random time, the expectation of the \supw, which is the interval of
time between the last and the next successful \cass
(see Figure~\ref{fig:seq}).

The stable system that we observe is the \ps: threads are entering it
(after exiting a successful \rl) at an average rate, stay inside,
then leave (while entering a new \rl).
The average number of threads inside the \ps is $\atps = \ct - \atrl$,
each thread stays for an average duration of \pw, and in average, one
thread is exiting the \rl every \supw \avsupe{\atrl}, by definition of
the \supw.
According to Little's law~\cite{littles-law}, we have:
\pr{
\begin{equation}
\label{eq:little-gen}
\atps = \pw \times 1 / \avsupe{\atrl}, \text{ \ie} \quad  \avsupe{\atrl} = \pw / (\ct - \atrl)
\end{equation}
}
{
\[ \atps = \pw \times \frac{1}{\avsupe{\atrl}}, \text{ \ie} \]
\begin{equation}
\label{eq:little-gen}
\frac{1}{\pw} \times \avsupe{\atrl} = \frac{1}{\ct - \atrl}
\end{equation}
}

\pr{We decompose a \supw into two parts: \watiw and \compw (as explained
in Section~\ref{sec:fra-exe}). 
We express the expectation of the \supw time as}%
{As explained in Section~\ref{sec:fra-exe}, we further decompose a \supw
into two parts, separated by the first access to the \ds after a
successful \cas. We can then write the average \supw as the sum of:
(i) the expected time before some thread starts its
  \acc (the \watiw), and
(ii) the expected \compw.
We compute these two expectations independently and gather them into
the \supw thanks to:}
\begin{equation}
\label{eq:little-sp}
\avsupe{\atrl} = \avwati{\atrl} + \rwh{\atrl}.
\end{equation}

When the \ds is contended, a thread is likely to be
successful after some failed \res. Therefore a thread that is
successful was  already in the \rl when the previous
successful \cas occurred.
\rr{This implies that the \acc to the \ds will
be due to a failed \cas, instead of a \rf.}%
The time before a thread starts its \acc is then the time before a
thread finishes its current \cww since there is a thread
currently executing a \cas.

\smallskip

\subsubsection{Expected \COmpw}
\label{sec:little-expa}

Since the \ds is contended, numerous threads are inside the \rl, and,
due to hardware conflicts, a \re can experience expansion:
the more threads inside the \rl, the longer time between a \cas
request and the actual execution of this \cas. The expectation of the
\compw can be written as:
\begin{equation}
\label{eq:little-ret}
\rwh{\atrl} = \scas + \cw + \avexp{\atrl} + \scas ,
\end{equation}
where \avexp{\atrl} is the expectation of expansion when there are \atrl
threads inside the \rl, in expectation.
This expansion can be computed in the same way as
in~\cite{our-disc15}, through the following differential equation:\\
\pr{\begin{minipage}{.4\textwidth}
\begin{equation*}
\left\{
\begin{array}{lcl}
\difavexp{\atrl} &=& \fcas \times \dfrac{\frac{\fcas}{2} + \avexp{\atrl}}{ \scas +\calrl + \scas + \avexp{\atrl}},\\
\avexp{1} &=& 0
\end{array} \right.
\end{equation*}\end{minipage}\hfill\begin{minipage}{.5\textwidth}
by assuming that the expansion starts as soon as strictly more than 1
thread are in the \rl, in expectation.\end{minipage}}{
\begin{equation*}
\left\{
\begin{array}{lcl}
\difavexp{\atrl} &=& \fcas \times \dfrac{\frac{\fcas}{2} + 
\avexp{\atrl}}{ \scas +\calrl + \scas + \avexp{\atrl}}\\
\avexp{1} &=& 0
\end{array} \right.,
\end{equation*}
by assuming that the expansion starts as soon as strictly more than 1
thread are in the \rl, in expectation.
}

\vspace*{.3cm}
\subsubsection{Expected \WAtiw}
\label{sec:litt-slack}

Concerning the \watiw, we consider that, at any time, the threads that
are running the \rl have the same probability to be anywhere in their
current \re. However, when a thread is currently executing a \cas, the
other threads cannot execute as well a \cas. The other threads are
thus in their \cww or expansion. For every thread, the time
before accessing the \ds is then uniformly distributed
between $0$ and $\cw+\avexp{\atrl}$.
\pr{Using a well-known formula on the expectation of the minimum of
  uniformly distributed random variables, we show
  in Appendix~\ref{app:lemsl} that
}
{

According to Lemma~\ref{lem:unif-min}, we conclude that
}
\begin{equation}
\label{eq:slack-cont}
\pp{\hspace{5cm}}\avwati{\atrl} = \left( \cw + \avexp{\atrl} \right) / (\atrl +1).
\end{equation}

\rr{\lemsl}

\vspace*{.1cm}
\subsubsection{Expected \SUpw}

We just have to combine Equations~\ref{eq:little-sp},
\ref{eq:little-ret}, and~\ref{eq:slack-cont} to obtain the general
expression of the expected \supw under contention: \pr{$\avsupe{\atrl} = \left( 1 + 1/(\atrl +1) \right) \left( \cw + \avexp{\atrl} \right) + 2\scas$,}%
{\[\avsupe{\atrl} = \left( 1 + \frac{1}{\atrl +1} \right) \left( \cw + \avexp{\atrl} \right) + 2\scas, \]}
which leads, according to Equation~\ref{eq:little-gen}, to
\begin{equation}
\label{eq:little-co}
\frac{1}{\pw} \times \left(
\frac{\atrl +2}{\atrl +1}  \left( \cw + \avexp{\atrl} \right)
+ 2\scas
\right) = \frac{1}{\ct - \atrl}.
\end{equation}

\subsection{Non-contended System}

When the system is not contended, logical conflicts are not likely to
happen, hence each thread succeeds in its \rl at its first {\it
  \re}. \Afort, no hardware conflict occurs. Each thread still
performs one success every \wl, and the \supw is given by
\pr{%
$\avsupe{\atrl} = (\pw + \mem+\cw+\scas)/\ct$.
Moreover, a thread spends in average \pw units of time
in the \rl within each \wl. As this holds for
every thread, we deduce
$\ct - \atrl = \atps = \pw/(\pw + \mem+\cw+\scas) \times \ct$.
Combining the two previous equations, we obtain
\begin{equation}
\label{eq:little-nc}
\frac{\avsupe{\atrl}}{\pw}  = \frac{1}{\ct - \atrl},
\text{ where } \avsupe{\atrl} = \frac{\mem+\cw+\scas}{\atrl}.
\end{equation}
}%
{
\begin{equation}
\label{eq:little-avsp}
\avsupe{\atrl} = \frac{\pw + \mem+\cw+\scas}{\ct}.
\end{equation}

Moreover, a thread spends in average $\mem+\cw+\scas$ units of time
in the \rl within each \wl. As this holds for
every thread, we can obtain the following expression for the total
average number of threads inside the \rl:
\begin{equation}
\label{eq:little-trl}
\atrl  = \frac{\mem+\cw+\scas}{\pw + \mem+\cw+\scas}  \times \ct = \frac{\mem+\cw+\scas}{\avsupe{\atrl}}
\end{equation}

Equation~\ref{eq:little-avsp} also gives $ \mem+\cw+\scas = \ct \times
\avsupe{\atrl} - \pw$, hence, thanks to Equation~\ref{eq:little-trl},
\begin{equation}
\label{eq:little-nc}
\atrl = \frac{\ct \times \avsupe{\atrl}-\pw}{\avsupe{\atrl}}, \text{ \ie} \quad
\frac{\avsupe{\atrl}}{\pw}  = \frac{1}{\ct - \atrl},
\end{equation}
where $\avsupe{\atrl} = \frac{\mem+\cw+\scas}{\atrl}$.}

\subsection{Unified Solving}

\pp{

\proofswitch

We show in the following theorem how to compute the throughput estimate; the proof, 
presented in~\ref{app:fixed-point},
manipulates equations in order to be
able to use the fixed-point Knaster-Tarski theorem.

}

\rr{
\proofswitch

We show in the following theorem how to compute the throughput estimate; the proof
manipulates equations in order to be
able to use the fixed-point Knaster-Tarski theorem.
} 

\begin{theorem}
\label{th:fixed-point}
The throughput can be obtained iteratively through a fixed-point
search, as $\thru = \left( \avsupe{\lim_{n \rightarrow \pinf} u_n} \right) ^{-1}$, where
\[ \left\{\begin{array}{ll}
u_0 = \frac{\mem + \cw + \scas}{\pw + \mem + \cw + \scas} \times \ct &\\
u_{n+1} = \frac{u_n \avsupe{u_n}}{\pw + u_n \avsupe{u_n}} \times \ct & \quad \text{for all } n \geq 0.
\end{array}\right.
\]
\end{theorem}
\rr{ 
\begin{proof}
\prooffp
\end{proof}
} 

\section{Constructive Approach}
\label{sec:cons}


In this section, we instantiate the probability distribution of the
\pww with an exponential distribution. We have therefore a
better knowledge of the behavior of the algorithm, particularly in
medium contention cases, which allows us to follow a fine-grained
approach that studies individually each successful operation together
with every \cas occurrence. We provide an elegant and efficient
solution that relies on a Markov chain analysis.

\vspace{-.3cm}
\subsection{Process}
\label{sec:mark-proc}

We have seen in Section~\ref{sec:fra-app-con} that
we split the contention domain into three modes:
no contention, medium contention or high
contention.
\pr{We }{The main idea is to }start from a configuration with a given
number of threads \trl \rr{just }after a successful \cas, and describe
what will happen until the next successful \cas: what will be the mode
of the next \supw, and\rr{ even} more precisely, which will be the
number of threads at the beginning of the next \supw.

As a basis, we consider the execution that would occur without any
other thread exiting the \ps (then entering the \rl); we call this
execution the {\it internal execution}. This execution follows the
\supw pattern described in Figure~\ref{fig:seq} (with an infinite
\watiw if the system is not contended).
On top of this basic \supw, we inject the threads that can exit the
\ps, which has a double impact. On the one hand, they increase the
number of threads inside the \rl for the next \supw. On the other
hand, if the first thread that exits the \ps starts its \re during the
\watiw of the \supw of the internal execution, then this thread will
succeed its \acc, which is a \rf, and will shrink the actual \watiw
of the current \supw.

According to the distribution probability of the arrival of the new
threads, we can compute the probability for the next \supw to
start with any number of threads. The expression of this stochastic
sequence of \supws in terms of Markov chains results in the throughput
{estimate}.

\vspp{-.4}
\subsection{Expansion}
\label{sec:mark-expa}

The expansion, as before, represents the additional time in the
execution time of a \re, due to the serialization of atomic
primitives. However, in contrary to Section~\ref{sec:little-expa}, we
compute here this additional time in the current \supw, according to
the number of threads \trl inside the \rl at the beginning of the
\supw.
The expansion only appears when the \supw is highly contended, \ie
when we can find a continuous sequence of \cass all through the
\supw.

The expansion is highly correlated with the way the cache coherence
protocol handles the exchange of cache lines between threads. We rely
on the experiments of the research report associated
with~\cite{ali-same}, which show that if several threads request for
the same cache line in order to operate a \cas, while another thread
is currently executing a \cas, they all have an equal probability to
obtain the cache line when the current \cas is over.

\def\herec{.5}
\def\wirec{1}
\def\wicw{2.5}
\newcommand{\dcasf}[2]{
\draw[pattern=north west lines, draw=none] (0,-#2*\herec) rectangle (#1*\wirec,-#2*\herec+\herec);
\draw[fill=red] (#1*\wirec,-#2*\herec) rectangle ++(\wirec,\herec) node[midway, align=center] {\cas};
}
\newcommand{\dcass}[2]{
\draw[fill=\green] (#1*\wirec,-#2*\herec) rectangle ++(\wirec,\herec) node[midway, align=center] {\cas};
}
\newcommand{\dcw}[2]{
\draw[fill=\maroon] (#1*\wirec,-#2*\herec) rectangle ++(\wicw,\herec) node[midway, align=center] {\cw};
}
\newcommand{\dpw}[3]{
\draw[fill=\grey] (#1*\wirec,-#2*\herec) rectangle ++(#3*\wirec,\herec) node[midway, align=center] {\pw};
}
\newcommand{\dwt}[3]{
\draw[densely dashed] (#1*\wirec,-#2*\herec+.5*\herec) -- ++(#3*\wirec,0);
}

\def\marup{1}
\def\decx{.2}
\def\decy{.35}

\newcommand\arcod[2]{
\draw[draw=blue,very thick,fill=blue] #1 -- ++ (0,-#2) -- ++(\decx,\decy) -- ++(-\decx,0);
}
\newcommand\arcou[2]{
\draw[draw=blue, very thick,fill=blue] #1 --  ++ (0,#2) -- ++(\decx,-\decy) -- ++(-\decx,0);}

\pp{
\begin{figure}
\begin{center}
\hspace*{-0cm}\begin{minipage}{.4\textwidth}
\begin{center}
\begin{tikzpicture}[scale=.73]
\clip (-1*\wirec,1*\herec) rectangle (8*\wirec,-9.5*\herec);

\dcass{0}{0}\dpw{1}{0}{8}
\dcasf{1}{3}\dcw{2}{3}
\dcasf{2}{2}\dcw{3}{2}
\dcasf{3}{5}\dcw{4}{5}
\dcasf{4}{4}\dcw{5}{4}
\dcasf{5}{1}\dcw{6}{1}
\dcasf{6}{6}\dcw{7}{6}

\dwt{2+\wicw}{3}{3.5}
\dwt{3+\wicw}{2}{1.5}\dcass{3+1.5+\wicw}{2}
\dwt{4+\wicw}{5}{1.5}
\dwt{5+\wicw}{4}{.5}

\draw[blue,densely dotted, thick] (5*\wirec,0) -- ++(0,-6.4*\herec) -- ++(-.5*\wirec,0) -- ++(0,-.5*\herec) ;
\draw[blue,densely dotted, thick] (6*\wirec,0) -- ++(0,-6.9*\herec);
\draw[blue,densely dotted, thick] (7*\wirec,0) -- ++(0,-6.4*\herec) -- ++(.5*\wirec,0) -- ++(0,-.5*\herec) ;
\draw[blue,densely dotted, thick] (1*\wirec,0) -- ++(0,-6.9*\herec);

\node[text width=1cm, align=center] at (4.5*\wirec,-8*\herec) { {\color{red} $\trl-4$}\\ {\color{black} vs}\\ {\color{green}1}};
\node[text width=1cm, align=center] at (6*\wirec,-8*\herec) { {\color{red} $\trl-5$}\\ {\color{black} vs}\\ {\color{green}2}};
\node[text width=1cm, align=center] at (7.5*\wirec,-8*\herec) { {\color{red} $\trl-6$}\\ {\color{black} vs}\\ {\color{green}3}};

\node[text width = 3cm, align=center] at (1.2*\wirec,-7.8*\herec) {\trl threads inside\\ the \rl};

\node[draw=black,rounded corners=4,scale=.8] at (-.5*\wirec,0.5*\herec ) {\tiny Thread 1};
\node[draw=black,rounded corners=4,scale=.8] at (-.5*\wirec,-0.5*\herec) {\tiny Thread 2};
\node[draw=black,rounded corners=4,scale=.8] at (-.5*\wirec,-1.5*\herec) {\tiny Thread 3};
\node[draw=black,rounded corners=4,scale=.8] at (-.5*\wirec,-2.5*\herec) {\tiny Thread 4};
\node[draw=black,rounded corners=4,scale=.8] at (-.5*\wirec,-3.5*\herec) {\tiny Thread 5};
\node[draw=black,rounded corners=4,scale=.8] at (-.5*\wirec,-4.5*\herec) {\tiny Thread 6};
\node[draw=black,rounded corners=4,scale=.8] at (-.5*\wirec,-5.5*\herec) {\tiny Thread 7};

\end{tikzpicture}
\end{center}
\captionof{figure}{Highly-contended execution\label{fig:mark-expa}}
\end{minipage}%
\hspace*{.35cm}\begin{minipage}{.65\textwidth}
\begin{center}
\begin{tikzpicture}[%
it/.style={%
    rectangle,
    text width=11em,
    text centered,
    minimum height=\pr{2}{3}em,
    draw=black!50,
    scale=.7,
  }
]
\coordinate (O) at (0,0);
\pha{pcas}{O}{\cas}{\pr{2}{5}}{\green}
\pha{sla}{pcas}{\wati{i}}{\pr{5}{8}}{\grey}
\pha{acc}{sla}{\cas}{\pr{2}{4}}{red}
\pha{cri}{acc}{\cw}{\pr{3}{4}}{\maroon}
\pha{exp}{cri}{\reexp{i}}{\pr{4}{6}}{\grey}
\pha{fcas}{exp}{\cas}{\pr{2}{5}}{\green}
\coordinate (intsl) at ($(sla.north west)+(0,.4*\marup)$); 
\coordinate (intsr) at ($(sla.north east)+(0,.4*\marup)$);
\coordinate (inte) at ($(fcas.north east)+(0,.4*\marup)$); 
\draw [decorate,decoration={brace,amplitude=10pt}]
(intsl)  --  (intsr)
node [black,midway,yshift=15pt] {\scriptsize 0 new thread};

\draw [decorate,decoration={brace,amplitude=10pt}]
(intsr) -- (inte)
node [black,midway,yshift=15pt] {\scriptsize $k+1$ new threads};
\arcod{($(acc.north west)!.2! (acc.north east) + (0,.5*\marup)$)}{.4*\marup}
\arcod{($(exp.north west)!.1! (exp.north east) + (0,.5*\marup)$)}{.4*\marup}
\arcod{($(exp.north west)!.9! (exp.north east) + (0,.5*\marup)$)}{.4*\marup}

\coordinate (extsl) at ($(sla.south west)+(0,-.4*\marup)$);
\coordinate (extsr) at ($(sla.south east)+(0,-.4*\marup)$);
\draw [decorate,decoration={brace,mirror,amplitude=10pt},text width=11em, align=center] (extsl) -- (extsr)
node (caca) [black,midway,yshift=-13pt,xshift=26pt] {\scriptsize at least 1 new thread};

\coordinate (Ob) at ($(sla.south west)!.2!(sla.south) + (0,-1.5*\marup)$);
\pha{accb}{Ob}{\rf}{\pr{2}{4}}{yellow}
\pha{crib}{accb}{\cw}{\pr{3}{4}}{\maroon}
\pha{expb}{crib}{\reexp{i}}{\pr{4}{6}}{\grey}
\pha{fcasb}{expb}{\cas}{\pr{2}{5}}{\green}
\arcou{(accb.south west)}{1.65*\marup}
\arcou{($(crib.south west)!.2! (crib.south east) + (0,-.6*\marup)$)}{.4*\marup}
\arcou{($(expb.south west)!.9! (expb.south east) + (0,-.6*\marup)$)}{.4*\marup}
\draw[very thick, draw=blue] (accb.south west) -- (fcasb.south east) -- (fcasb.north east) -- (accb.north west);
\coordinate (extnl) at ($(accb.south west)+(0,-.5*\marup)$); 
\coordinate (extnr) at ($(fcasb.south east)+(0,-.5*\marup)$);
\draw [decorate,decoration={brace,mirror,amplitude=10pt},text width=11em, align=center]
(extnl) -- (extnr)
node [black,midway,yshift=-15pt] {\scriptsize $k$ new threads};

\draw[dotted, draw=black] (intsl) -- (extsl);
\draw[dotted, draw=black] (intsr) -- (extsr);
\draw[dotted, draw=black] (inte) -- (fcas.south east);
\draw[dotted, draw=black] (accb.north west) -- (extnl);
\draw[dotted, draw=black] (fcasb.north east) -- (extnr);
\node[text width=1.2cm, align=center] (inttext) at ($(pcas.west) + (-1*\marup,0)$) {Internal\\ execution};
\node[anchor=east] (eint) at ($(inttext.east) + (0.2,1.2*\marup)$) {\eve{int}};
\node[anchor=east] (eext) at ($(inttext.east) + (0.2,-2*\marup)$){\eve{ext}};

\path[->,out=90,in=-180] ($(inttext.north west)!.3!(inttext.north)$) edge (eint);
\path[->,out=-90,in=-180] ($(inttext.south west)!.3!(inttext.south)$) edge (eext);
\end{tikzpicture}
\end{center}
\captionof{figure}{Possible executions\label{fig:ex-eint-eext}}
\end{minipage}
\end{center}
\end{figure}
}

\rr{
\begin{figure}
\begin{center}
\begin{tikzpicture}[scale=1.6]
\clip (-1.2*\wirec,1*\herec) rectangle (8*\wirec,-9*\herec);

\dcass{0}{0}\dpw{1}{0}{8}
\dcasf{1}{3}\dcw{2}{3}
\dcasf{2}{2}\dcw{3}{2}
\dcasf{3}{5}\dcw{4}{5}
\dcasf{4}{4}\dcw{5}{4}
\dcasf{5}{1}\dcw{6}{1}
\dcasf{6}{6}\dcw{7}{6}

\dwt{2+\wicw}{3}{3.5}
\dwt{3+\wicw}{2}{1.5}\dcass{3+1.5+\wicw}{2}
\dwt{4+\wicw}{5}{1.5}
\dwt{5+\wicw}{4}{.5}

\draw[blue,densely dotted, thick] (5*\wirec,0) -- ++(0,-6.4*\herec) -- ++(-.5*\wirec,0) -- ++(0,-.5*\herec) ;
\draw[blue,densely dotted, thick] (6*\wirec,0) -- ++(0,-6.9*\herec);
\draw[blue,densely dotted, thick] (7*\wirec,0) -- ++(0,-6.4*\herec) -- ++(.5*\wirec,0) -- ++(0,-.5*\herec) ;
\draw[blue,densely dotted, thick] (1*\wirec,0) -- ++(0,-6.9*\herec);

\node[text width=2cm, align=center] at (4.5*\wirec,-8*\herec) { {\color{red} $\trl-4$}\\ {\color{black} vs}\\ {\color{green}1}};
\node[text width=2cm, align=center] at (6*\wirec,-8*\herec) { {\color{red} $\trl-5$}\\ {\color{black} vs}\\ {\color{green}2}};
\node[text width=2cm, align=center] at (7.5*\wirec,-8*\herec) { {\color{red} $\trl-6$}\\ {\color{black} vs}\\ {\color{green}3}};

\node[text width = 3cm, align=center] at (1.2*\wirec,-7.8*\herec) {\trl threads inside\\ the \rl};

\node[draw=black,rounded corners=4,scale=.8] at (-.6*\wirec,0.5*\herec ) {\small Thread 1};
\node[draw=black,rounded corners=4,scale=.8] at (-.6*\wirec,-0.5*\herec) {\small Thread 2};
\node[draw=black,rounded corners=4,scale=.8] at (-.6*\wirec,-1.5*\herec) {\small Thread 3};
\node[draw=black,rounded corners=4,scale=.8] at (-.6*\wirec,-2.5*\herec) {\small Thread 4};
\node[draw=black,rounded corners=4,scale=.8] at (-.6*\wirec,-3.5*\herec) {\small Thread 5};
\node[draw=black,rounded corners=4,scale=.8] at (-.6*\wirec,-4.5*\herec) {\small Thread 6};
\node[draw=black,rounded corners=4,scale=.8] at (-.6*\wirec,-5.5*\herec) {\small Thread 7};

\end{tikzpicture}
\end{center}
\caption{Highly-contended execution\label{fig:mark-expa}}
\end{figure}
}

We draw an illustrative example in Figure~\ref{fig:mark-expa}. The
green \cass are successful while the red \cass fail. To lighten
the picture, we hide what happened for the threads before they
experience a failed \cas. The horizontal dash lines represent the time
where a thread wants to access the data in order to operate a \cas
but has to wait because another thread owns the data in exclusive
mode.
We can observe in this example that the first thread that accesses
the \ds is not the thread whose operation returns.

We are given that \trl threads are inside the \rl at the end of the
previous successful \cas, and we only consider those threads. When
such a thread executes a \cas for the first time, this \cas is
unsuccessful. The thread was in the \rl when the successful \cas has
been executed, so it has read a value that is not up-to-date
anymore. However, this failed \cas will bring the current version of
the value (to compare-and-swap) to the thread, a value that will be
up-to-date until a successful \cas occurs.

So we have firstly a sequence of failed \cass until the first thread
that operated its \cas within the current \supw finishes its \cww. At
this point, there exists a thread that is executing a \cas. When this
\cas is finished, some threads compete to obtain the cache line. We
have two bags of competing threads: in the first bag, the thread that
just ended its \cww is alone, while in the second bag, there are all
the threads that were in the \rl at the beginning of the \supw, and
did not operate a \cas yet. The other, non-competing, threads are
running their \cww and do not yet want to access the data.

As described before, every thread has the same probability to become the next
owner of the cache line. If a thread from the first bag is drawn, then
the \cas will be successful and the \supw ends. Otherwise, the \cas is
a failure, and we iterate at the end of this failed \cas. However, the
thread that just failed its \cas is now executing its \cww,
and does not request for a new \cas until this work has been done,
thus it is not anymore in the second bag. In addition, the thread that
had executed its \cas after the thread of the first bag is now back
from its \cww and falls into the first bag. The process iterates until
a thread is drawn from the first bag.

As a remark, note that we do not consider threads that are not in the
\rl at the beginning of the \supw since even if they come back from
the \ps during the \supw, their \rf will be delayed and their \cas is
likely to occur after the end of the \supw.

Theorem~\ref{th:mark-expa} \pp{, proved in Appendix~\ref{app:aliexp}, }
 gives the explicit formula for the expansion\pr{.}{, based on the previous
explanations.}

\newcommand{\prsu}[1]{\ema{p_{#1}}}

\begin{theorem}
\label{th:mark-expa}
The expected time between the end of the \cww of the first
thread that operates a \cas in the \supw and the beginning of a
successful \cas is given by:\vspp{0}
\[ \reexp{\trl} = \lceil\cw/\scas\rceil\scas - \cw +
\sum_{i=1}^{\tcom} \frac{i(i-1)}{\left(\tcom\right)^i} \frac{(\tcom-1)!}{(\tcom-i)!} \times \scas, \pp{\quad\text{where }\tcom = \trl - \lceil\cw/\scas\rceil +1.}\]
\rr{where $\tcom = \trl - \lceil\cw/\scas\rceil +1$.}
\end{theorem}
\rr{ 
\begin{proof}
\proofaliexp
\end{proof}
}  

\subsection{Formalization}

The \pww follows an exponential distribution, whose mean is
\pw. More precisely, if a thread starts a \ps at the instant $t_1$,
the probability distribution of the execution time of the \ps is
\pr{$ t \mapsto \lambda \expu{-\lambda (t-t_1)} \indi{[t_1,\pinf[}{t}$, where $\lambda = 1/\pw. $}
{\[ t \mapsto \lambda \expu{-\lambda (t-t_1)} \indi{[t_1,\pinf[}{t}, \text{ where } \lambda = \frac{1}{\pw}. \]}
This probability distribution is memoryless, which implies that the
threads that are executing their \ps cannot be differentiated: at a
given instant, the probability distribution of the remaining execution
time is the same for all threads in the \ps, regardless of when the \ps began. For all
threads, it is defined by:
\pr{ $ t \mapsto \lambda \expi{-\lambda t}$, where $\lambda = 1/\pw. $}
{\[ t \mapsto \lambda \expu{-\lambda t}, \text{ where } \lambda = \frac{1}{\pw}. \]}

For the behavior in the \rl, we rely on the same approximation as in
the previous section, \ie when a successful thread exits its
\rl, the remaining execution time of the \re of every other thread
that is still in the \rl is uniformly distributed between $0$ and the
execution time of a whole \re. We have seen that the expectation of
this remaining time is the size of the execution time of a \re divided
by the number of threads inside the \rl plus one. Here, we assume that
a thread will start a \re at this time.
This implies another kind of memoryless property: the behavior of a
thread that is in the \rl does not depend on the moment that it
entered its \rl.

To tackle the problem of estimating the throughput of such a system,
we use an approach based on Markov chains. We study the behavior of
the system over time, step by step: a state of the Markov chain
represents the state of the system when the current \supw began (\ie
just after a successful \cas) and (thus) the system changes state at
the end of every successful \cas.
According to the current state, we are able to compute the probability
to reach any other state at the beginning of the next \supw.
In addition, the two memoryless properties render the description of a
state easy to achieve: the number of threads inside the \rl when the
current success begins, indeed fully characterizes the system.

We recall that \trl is the number of threads inside the \rl when the
\supw begins.
The Markov chain is strongly connected with \trl, since it is composed of
\ct states $\sta{0}, \sta{1}, \dots, \sta{\ct-1}$, where, for all $i
\in \inte{\ct-1}$, the \supw is in state \sta{i} iff $\trl=i$. For all
$(i,j) \in \inte{\ct-1}^2$, $\pro{\sta{i} \rightarrow \sta{j}}$
denotes the probability that a success characterized by \sta{j}
follows a success in state \sta{i}.
$\wati{\sta{i} \rightarrow \sta{j}}$ denotes the
\watiw that passed while the system has gone from state \sta{i} to
state \sta{j}.
This \watiw can be expressed based on the \watiw \wati{i} of the
internal execution, \ie the execution that involves only the $i$
threads of the \rl and ignores the other threads (see
Section~\ref{sec:mark-proc}). \rr{Recall that we consider that the \watiw of
the internal execution with $0$ thread is infinite, since no thread will
access the \ds.}
In the same way, we denote by \rw{i} the \compw of the internal
execution, hence $\rw{i} = \cc + \cw + \reexp{i} + \cc$.

\newcommand{\intgen}[1]{\ema{\mathcal{I}_{\mathrm{#1}}}}
\newcommand{\intnoc}{\intgen{noc}}
\newcommand{\intmid}{\intgen{mid}}
\newcommand{\inthig}{\intgen{hi}}

\newcommand{\ihig}{\ema{i_{\mathrm{hi}}}}

We have seen that the level of contention (mode) is determined by
\trl, hence the interval \inte{\ct-1} can be partitioned into
\pr{$\inte{\ct-1} = \intnoc \cup \intmid \cup \inthig$,}
{\[ \inte{\ct-1} = \intnoc \cup \intmid \cup \inthig, \]}
where the partitions correspond to the different contention levels.
So, by definition, $\intnoc = \{ 0 \}$, and for all $i \in \intnoc
\cup \intmid$, $\reexp{i} = 0$ (see Section~\ref{sec:fra-app-con}).

The \supw is highly-contended, \ie we have a continuous sequence of
\cass in the \supw, if the sum of the execution time of all the
\cass that need to be operated exceeds the \cww. Hence
$\inthig = \inte[\ihig]{\ct-1}$, where
\pr{$\ihig = \min \{ i \in \inte[1]{\ct-1} \; | \; i \times \scas > \cw  \}$.}
{\[ \ihig = \min \{ i \in \inte[1]{\ct-1} \; | \; i \times \scas > \cw  \}. \]}
In addition, as the sequence of \cass is continuous when the
contention is high, the \watiw is null when the \supw is highly
contended, \ie, for all $i \in \inthig$, $\wati{i} = 0$, and \afort,
$\wati{\sta{i} \rightarrow \sta{\star}} = 0$.

Otherwise, the \supw is in medium contention, hence $\intmid =
\inte[1]{\ihig-1}$. Moreover, if $i \in \intmid$, $\wati{i} > 0$, and
$\reexp{i} = 0$, because the \cass synchronized during the previous
\supw and will not collide any more in the current \supw.

\pp{Everything is now in place to be able to obtain the stationary
  distribution of the Markov chain, and in turn to compute the
  throughput and the failure rate estimates. The reasoning that leads
  to the computation of the probability of going from state \sta{i} to
  state \sta{i+k} can be roughly summarized by
  Figure~\ref{fig:ex-eint-eext}, where we start from an internal
  execution with $i$ threads inside the \rl and the blue arrows
  represent the threads that exit the \ps. Two non-overlapping events
  can then potentially occur: either (event \eve{ext}) the first thread exiting the \ps
  starts within $[0,\wati{i}[$, \ie in the \watiw of
      the internal execution, or (event \eve{int}) the first thread entering the \rl
      starts after $t=\wati{i}$. The details can be
      found in Appendix~\ref{app:markov}.
}

\rr{\wholemarkov}

\section{Experiments}
\label{sec:xp}

To validate our analysis results, we use two main types of lock-free
algorithms.  
In the first place, we consider a set of basic algorithms where operations 
can be completed with a single successful \cas.
This set of algorithms includes: (i) synthetic
designs, that cover the design space of possible lock-free data structures;
(ii) several fundamental designs of \ds operations such as lock-free
stacks~\cite{lf-stack} (\popop, \pushop),
queues~\cite{lf-queue-michael} (\deqop), counters~\cite{count-moir}
(\incop, \decop).
As a second step, we consider more advanced lock-free operations that
involve helping mechanisms, and show how to use our analysis in this
context.  Finally, in order to highlight the benefits of the analysis
framework, we show how it can be applied to i) determine a beneficial back-off
strategy and ii) optimize the memory management scheme used by a \ds,
in the context of an application.

We also give insights about the strengths of our two approaches.
\pr{The }{On the one hand, the }constructive approach exhibits better predictions
due to the tight estimation of the failing retries. On the other hand, the
\avba approach is applicable to a broader spectrum of algorithmic designs
as it leaves room to abstract complicated algorithmic designs.

\vspp{-.2}
\subsection{Setting}

We have conducted experiments on an Intel ccNUMA workstation
system. The system is composed of two sockets equipped with
Intel Xeon E5-2687W v2 CPUs\pp{.}\rr{ with frequency band
  \ghz{1.2\text{-}3.4.} The physical cores have private L1, L2 caches
  and they share an L3 cache, which is \megb{25}.}  In a socket, the
ring interconnect provides L3 cache accesses and core-to-core
communication. \rr{Due to the bi-directionality of the ring
  interconnect, uncontended latencies for intra-socket communication
  between cores do not show significant variability.}\pp{Threads are
  pinned to a single socket to minimize non-uniformity in \rf and \cas
  latencies. }%
\pp{The methodology in~\cite{david-emp-atom} is used to measure the
  \cas and \rf latencies, while the \pww is implemented
  by a for-loop of {\it Pause} instructions. }%
\rr{Our model assumes uniformity in the \cas and \rf latencies on the
  shared cache line. Thus, threads are pinned to a single socket to
  minimize non-uniformity in \rf and \cas latencies.  In the
  experiments, we vary the number of threads between 4 and 8 since the
  maximum number of threads that can be used in the experiments are
  bounded by the number of physical cores that reside in one socket. }%
We show the experimental results with 8 threads.

In all figures, the y-axis shows both the throughput\rr{ values}, \ie
number of operations completed per second, and the ratio of failing to
successful retries (multiplied by $10^6$, for readability),
while the mean of the exponentially distributed \pww \pw is
represented on the x-axis. 
The number of failures per success in the \avba
approach is computed as $\atrl-1$ and \pr{in the constructive approach
by stochastically counting the failing \cass inside a \supw
(see Appendix~\ref{sec:nbf}).}{ is described in Section~\ref{sec:nbf}
for the constructive approach.}

\rr{\figsynthcst}

We have also added a straightforward upper bound as a baseline
approach,\rr{ which is} defined as the minimum of $1/(\rc+\cw+\cc)$ (two
successful \res cannot overlap) and $\ctot/(\pw+\rc+\cw+\cc)$ (a
thread can succeed only once in each \wl).

\vspp{-.2}
\subsection{Basic Data Structures}

\pr{Firstly, we consider lock-free operations that can be completed
with a single successful \cas.
  We provide predictions, on the one
  hand, on a set of synthetic tests that have been constructed to
  abstract different possible design patterns of lock-free data
  structures (value of \cw) and different application contexts (value
  of \pw), and, on the other hand, on the well-known Treiber's
  stack. The results are depicted in Appendix~\ref{app:xp-basic}.

}
{
\synthtreib

\figsynthpoi
\figsynthconst

\subsubsection{Synthetic Tests}
\label{sec:synt_tests}
\synth

\subsubsection{Treiber's Stack}

\figtreib

\treib

}

\subsection{Towards Advanced Data Structure Designs}
\label{sec:advanced-ds}

Advanced lock-free operations generally require multiple pointer
updates that cannot be done with a single \cas. 
One way to design such operations, in a lock-free manner, is to
use helping mechanisms: an inconsistency will be fixed eventually by
some thread.
Here we consider two \dss that apply immediate helping, the queue
from~\cite{lf-queue-michael} and the deque designed in~\cite{deq}. In the queue
experiment (Figure~\ref{fig:enqueue}), we run the \enqop operation on the
queue with and without memory management; in the \deq experiment, each
thread is dedicated to an end of the \deq (equally distributed), while
we vary the proportion of push operations (colors in Figure~\ref{fig:deq}).
\pp{More details about the implementations and the throughput estimate
obtained through a slight modification of the \avba approach can be
found in Appendix~\ref{app:ads}.}

\pp{\figsidebyside{h!}{.5}{enqueue_pp_disc}{Enqueue on MS queue}{enqueue}{.5}{deq_pp}{Operations on \deq}{deq}}

\rr{
\sumads

\subsubsection{Expected Expansion for the Advanced Data Structures}
\fulladsexp

\subsubsection{Expected Slack Time for the Advanced Data Structures}
\fulladswati

\subsubsection{Enqueue on Michael-Scott Queue}
\fullenq

\subsubsection{\Deq}
\label{sec:xp-deq}
\fulldeq}

\subsection{Applications}

\pr{\figsidebyside{b!}{.5}{back-offs-new}{Performance impact of our back-off tunings}{fig:bos}{.5}{adaptive_pp_disc}{Adaptive MM with varying mean \pw}{fig:mm_adaptive}} {

\begin{figure}[h!]
\includegraphics[width=\textwidth]{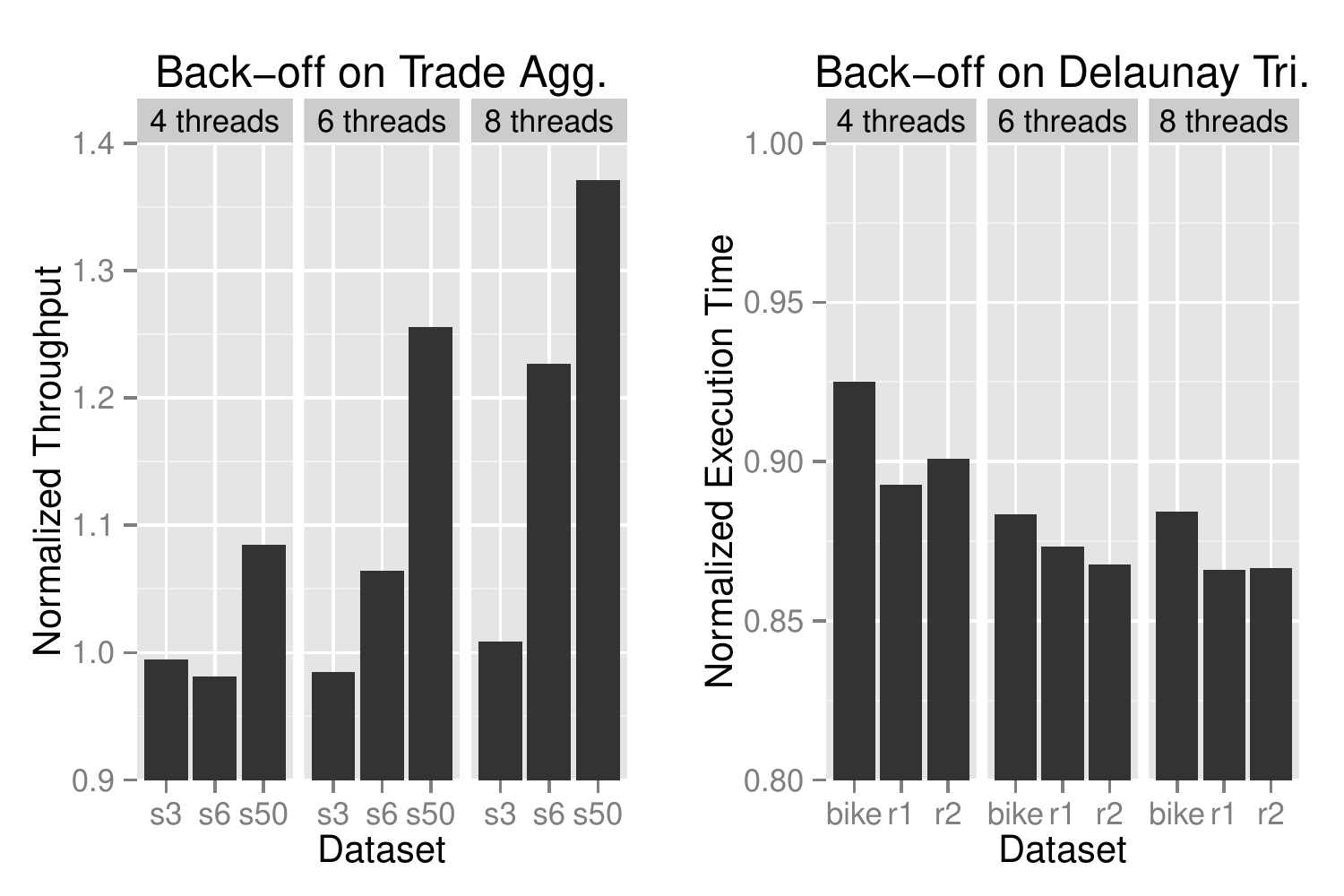}
\captionof{figure}{Performance impact of our back-off tunings \label{fig:bos}}
\end{figure}

}

\subsubsection{Back-off Optimizations}

When the \pww is known, we can deduce from our analysis a simple and
efficient back-off strategy: as we are able to estimate the value for
which the throughput is maximum, we just have to back-off for the time
difference between the peak \pw and the actual \pw.
\pr{In Appendix~\ref{app:bo}, we compare this back-off strategy
against widely known strategies, namely exponential and linear, on a
synthetic workload. }%
{\fullbosy}%
In Figure~\ref{fig:bos}, we apply our constant
back-off on a Delaunay triangulation application~\cite{caspar},
provided with several workloads. The application uses a stack in two
phases, whose first phase pushes elements on top of the stack without
delay. We are able to estimate a corresponding back-off time, and we
plot the results by normalizing the execution time of our back-offed
implementation with the execution time of the initial implementation.

A measure or an estimate of \pw is not always available (and could
change over time, see next section), therefore we propose also an
adaptive strategy: we incorporate in the \ds a monitoring routine that tracks
the number of failed \res, employing a sliding window. As our analysis
computes an estimate of the number of failed \res as a function of
\pw, we are able to estimate the current \pw, and hence the
corresponding back-off time like previously.

We test our adaptive back-off mechanism on a workload originated
from~\cite{taq-se}, where global operators of exchanges for financial
markets gather data of trades with a microsecond accuracy. We assume
that the data comes from several streams, each
of them being associated with a thread. All threads enqueue the
elements that they receive in a concurrent queue, so that they can be
later aggregated. We extract from the original data a trade stream
distribution that we use to generate similar streams that reach the
same thread; varying the number of streams to the same thread leads to
different workloads. The results, represented as the normalized
throughput (compared to the initial throughput) of trades that are enqueued when the
adaptive back-off is used, are plotted in Figure~\ref{fig:bos}.
For any number of threads, the queue is not
contended on workload s3, hence our improvement is either small or
slightly negative. On the contrary, the workload s50 contends the
queue and we achieve very significant improvement.

\rr{\begin{figure}[h!]
\begin{center}
\includegraphics[width=.7\textwidth]{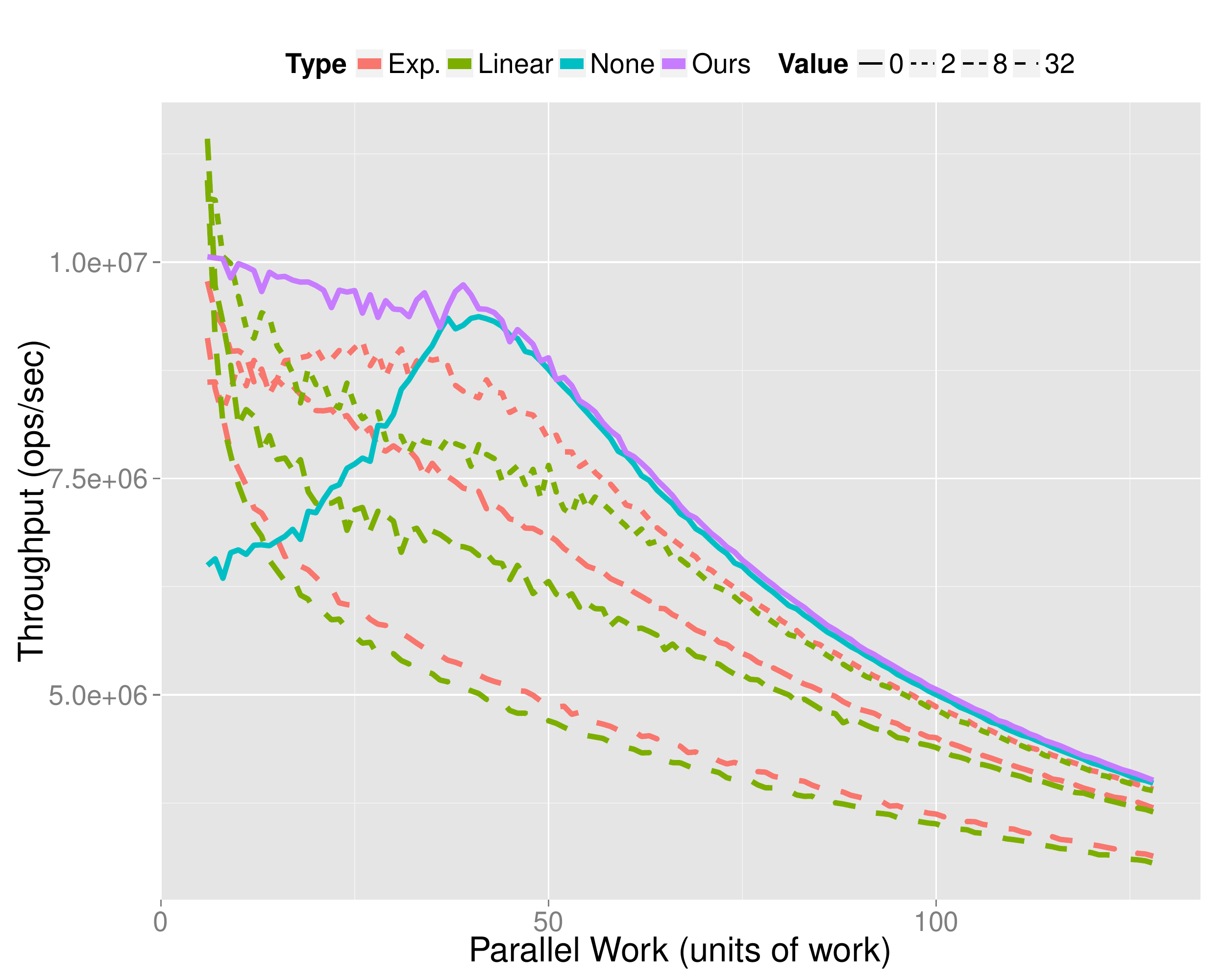}
\end{center}
\caption{Back-off Tuning on Treiber's Stack\label{fig:bo-synt}}
\end{figure}}

\subsubsection{Memory Management Optimization}

Memory Management (MM) is an inseparable part of dynamic concurrent
\dss. In contrary to lock-based implementations, a node that has been
{\it removed} from a lock-free \ds can still be accessed by other threads,
\eg if they have been delayed. Collective decisions are thus required
in order to {\it reclaim} a node in a safe manner.
A well-known solution to deal with this problem is the hazard pointers
technique~\cite{Mic04b}. \pp{In an implementation of such design each
  thread lists the nodes that it accesses and the nodes that it has
  removed. When the number of nodes it has removed reaches a threshold,
  it reclaims its listed removed nodes that are not listed as
  accessed by any thread.}

\newcommand\nodhp[1]{\ema{\mathcal{N}_{#1}}}
\newcommand\delhp[1]{\ema{\mathcal{D}_{#1}}}

\rr{A traditional design to implement this technique works as follows.
Each thread \thr{i}, maintains two lists of nodes: \nodhp{i} contains
the nodes that \thr{i} is currently accessing, and \delhp{i} stores
the nodes that have been removed from the \ds by \thr{i}. Once a
threshold on the size of \delhp{i} is reached, \thr{i} calls a routine
that: (i) collects the nodes that are accessed by any other thread, \ie
\nodhp{j} for $j \neq i$ (collection phase), and (ii) for each element
in \delhp{i}, checks whether someone is accessing the element, \ie
whether it belongs to $\cup_{j \neq i} \nodhp{j}$, and if not,
reclaims it (reclamation phase).}

The main goal of our adaptive MM scheme is to distribute this
extra-work in a way that the loss in performance is largely leveraged,
knowing that additional work can be an advantage under high-contention
(see previous section).
The optimization is based on two main
modifications.
\pr{First, we divide the reclamation phase of the traditional MM scheme
into quanta (equally-sized chunks), whose finer granularity allows for
accurate back-off times. Second, we track continuously the contention
level in the same way as our adaptive back-off. See Appendix~\ref{app:mm}.}%
{First, the granularity has to be finer, since the
additional quantum that the back-off mechanism uses, has to
be rather small (hundreds of cycles for a queue). Second, we need to
track the contention level on the \ds in order to be able to inject the work
at a proper execution point.}

\rr{\figcompmm}

\rr{\begin{figure}[t!]
\begin{center}
\includegraphics[width=.9\textwidth]{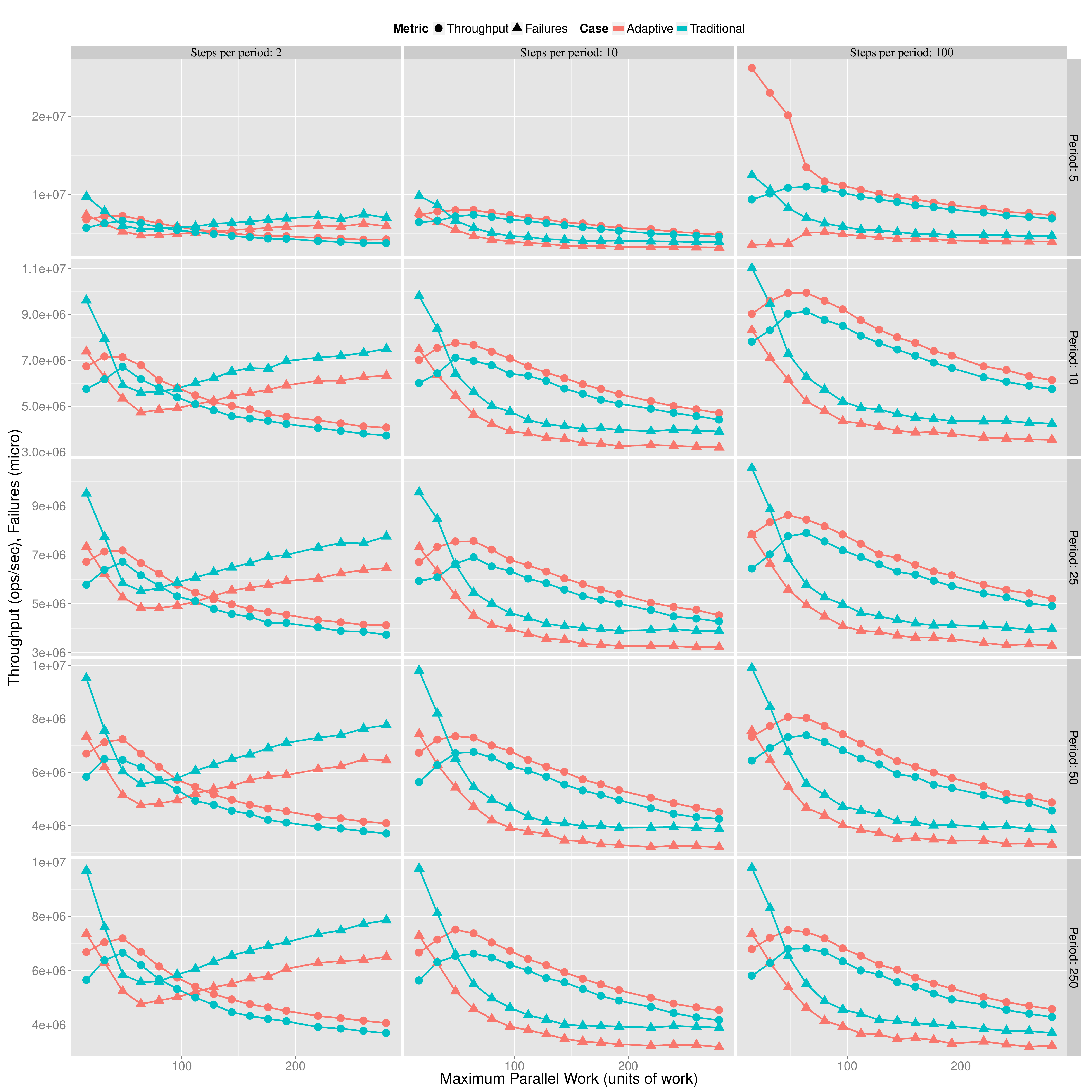}
\end{center}
\caption{Adaptive MM with varying mean \pw\label{fig:mm_adaptive}}
\end{figure}}

\rr{\bothmm}

\pr{%
We emulate the behavior of many scientific applications, that are
  built upon a pattern of alternating phases, that are
  communication-intensive (synchronization phase) or
  computation-intensive. Here we assume a synchronization ensured
  through a shared \ds, hence the communication-intensive phases
  correspond to a high access rate to the \ds, while \ds is accessed
  at a low rate during a computation-intensive phase. The \pww still
  follows an exponential distribution of mean \pw, but \pw varies in a
  sinusoidal manner with time.
To study also the impact of the continuity of the change in \pw, \pw
is set as a step approximation of a sine function. Thus, two
additional parameters rule the experiment: the period of the
oscillating function represents the length of the phases, and the
number of steps within a period depicts how continuous are the phase
changes.
}%
{\adaptsine}

In Figure~\ref{fig:mm_adaptive}, we compare our approach with the
traditional implementation for different periods of the sine function,
on the \deqop of the Michael-Scott queue~\cite{lf-queue-michael}.
The adaptive MM, that relies on the analysis presented
in this paper, outperforms the traditional MM
because it provides an advantage both under low contention due to the
costless (since delayed) invocation of the MM and under high
contention due to the back-off effect.

\pp{\vspace{-.2cm}}
\section{Conclusion}
\label{sec:conc}

In this paper we have presented two analyses for calculating the
performance of lock-free \dss in dynamic environments. The first
analysis has its roots in queuing theory, and gives the flexibility to
cover a large spectrum of configurations. The second analysis makes
use of Markov chains to exhibit a stochastic execution; it gives
better results, but it is restricted to simpler \dss and exponentially
distributed \pww.
\rr{We have evaluated the quality of the prediction on basic \dss like
stacks, as well as more advanced \dss like optimized queues and
\deqs. Our results can be directly used by algorithmicians to gain a
better understanding of the performance behavior of different designs,
and by experimentalists to rank implementations within a fair
framework. }%
We have\rr{ also} shown how to use our results to tune applications using
lock-free codes. These tuning methods include: (i) the calculation of
simple and efficient back-off strategies whose applicability is
illustrated in application contexts; (ii) a new adaptative memory
management mechanism that acclimates to a changing environment.

The main differences between the \dss of this paper and linked lists,
skip lists and trees occur when the size of the data structure
grows. With large sizes, the performance is dominated by the traversal
cost that is ruled by the cache parameters. The reduction in the size
of the data structure decreases the traversal cost which in turn
increases the probability of encountering an on-going \cas operation
that delays the threads which traverse the link.
The expansion, which can additionally be supported unfavorably by
helping mechanisms, appears then as the main performance degrading
factor.
While the analysis becomes easier for high degrees of parallelism
(large \ds size), being able to describe the behavior of lock-free
data structures as the degree of parallelism changes constitutes the
main challenge of our future work.

\pp{\small{\bibliography{./bibliography/shorthead,./bibliography/biblio}}}

\pp{\newpage\appendix

\section{Average-based Approach}
\label{app:avba}

\subsection{Helping Lemma}
\label{app:lemsl}
\lemsl

\subsection{Proof of Theorem~\ref{th:fixed-point}}
\label{app:fixed-point}

\prooffp

\section{Constructive Approach}
\label{app:markov}

\subsection{Proof of Theorem~\ref{th:mark-expa}}
\label{app:aliexp}
\proofaliexp

\wholemarkov
\figsynthcst

\section{Experiments}

\subsection{Basic Data Structures}
\label{app:xp-basic}
\synthtreib

\figsynthpoi
\figsynthconst

\paragraph{Synthetic Tests}

\synth

\paragraph{Treiber's Stack}

\figtreib

\treib

\subsection{Towards Advanced Data Structure Designs}
\label{app:ads}

\sumads

\paragraph{Expected Expansion for the Advanced Data Structures}
\fulladsexp

\paragraph{Expected Slack Time for the Advanced Data Structures}
\fulladswati

\paragraph{Enqueue on Michael-Scott Queue}
\fullenq

\paragraph{\Deq}
\label{sec:xp-deq}
\fulldeq

\subsection{Applications}

\paragraph{Back-off Optimizations}
\label{app:bo}

\begin{figure}[h!]
\begin{center}
\includegraphics[width=.6\textwidth]{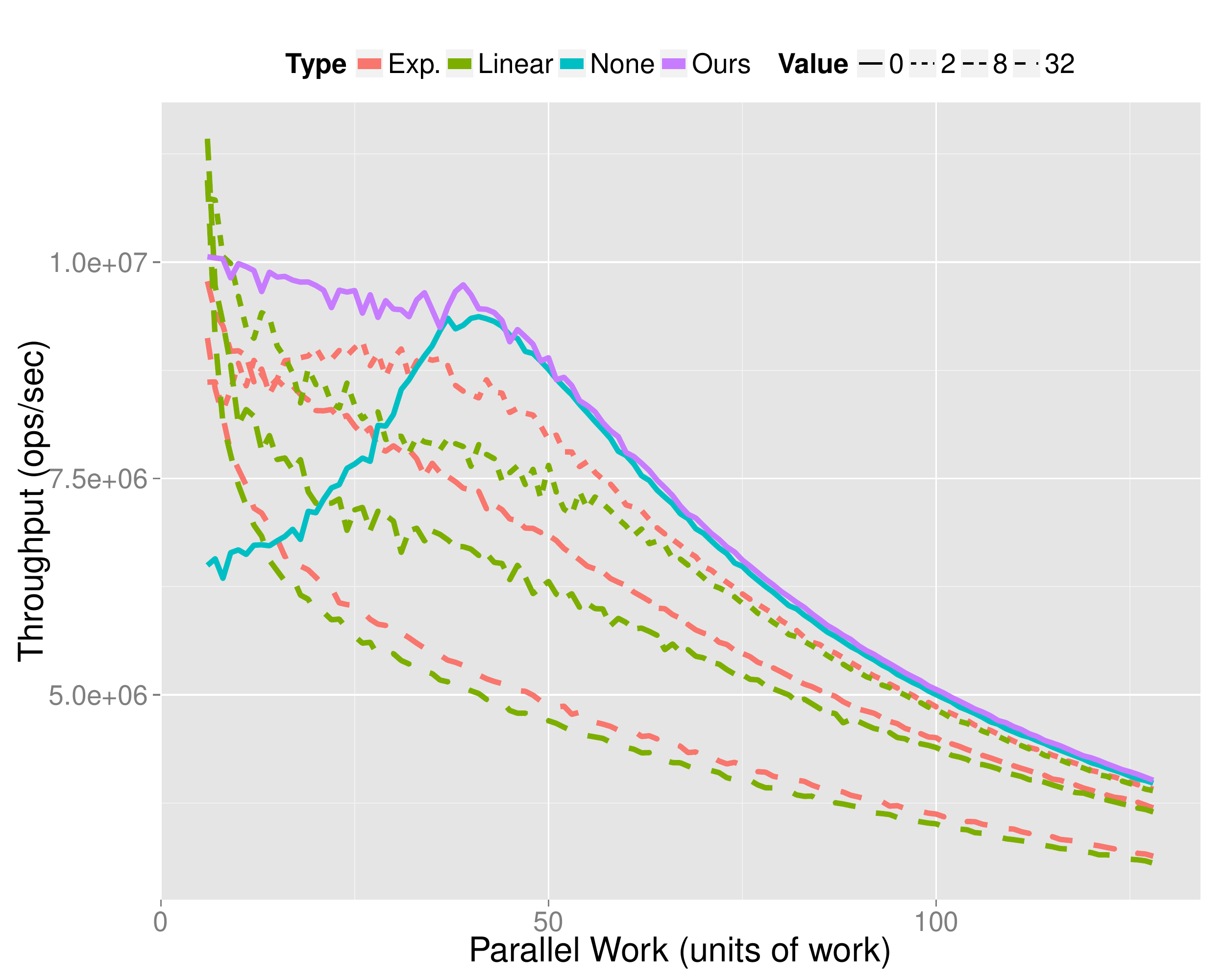}
\end{center}
\caption{Back-off Tuning on Treiber's Stack\label{fig:bo-synt}}
\end{figure}

\fullbosy

\figcompmm

\paragraph{Memory Management Optimization}
\label{app:mm}

\begin{figure}[b!]
\begin{center}
\includegraphics[width=.7\textwidth]{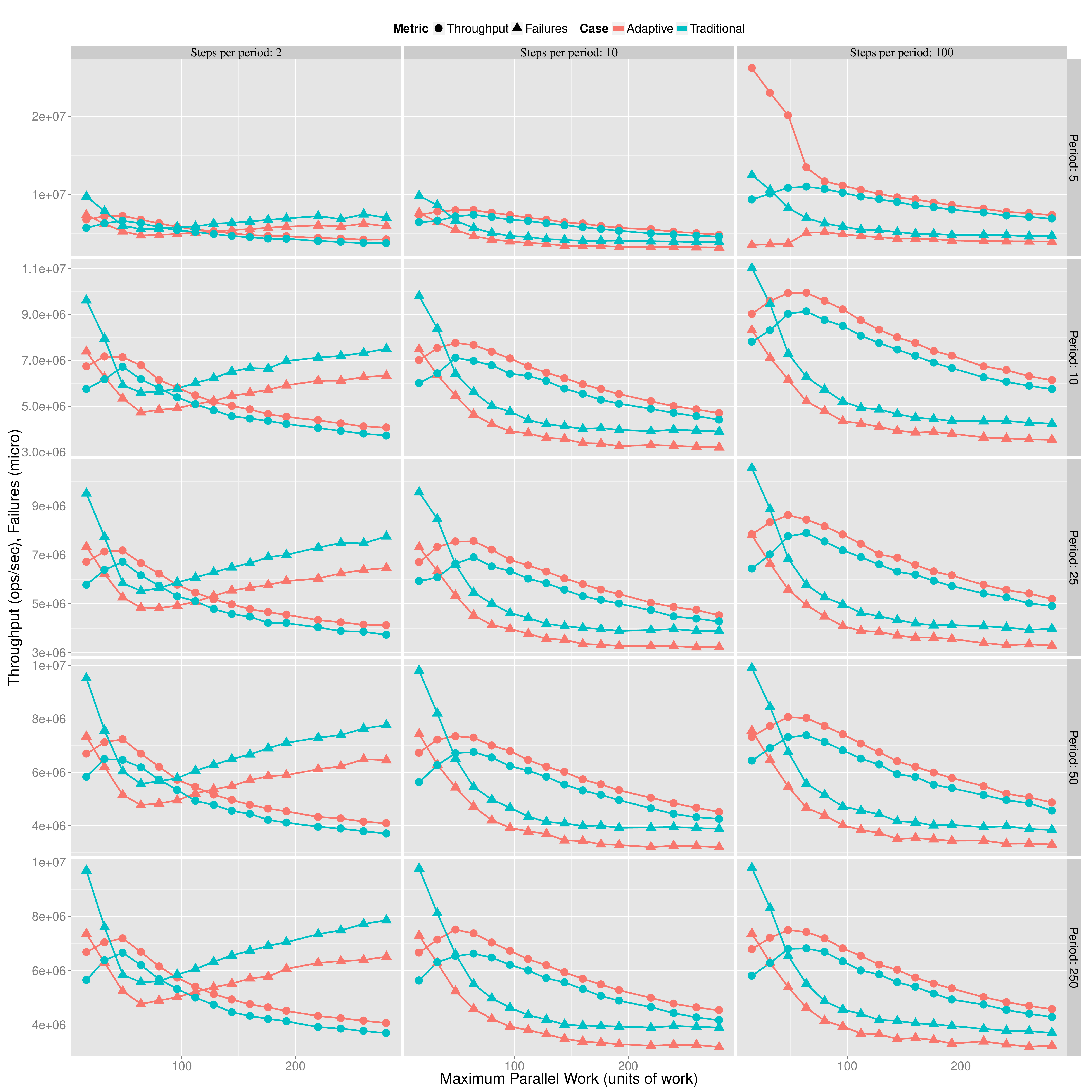}
\end{center}
\caption{Extensive experiments on adaptive MM\label{fig:ad-ext-mm}}
\end{figure}

~\\\bothmm

We show here extended experiments on Memory Management
optimization. Results are depicted in Figure~\ref{fig:ad-ext-mm}.

}

\rr{\bibliography{./bibliography/longhead,./bibliography/biblio}}

\end{document}